\documentclass[11pt]{article}
\usepackage{amsmath,amssymb,amsfonts,amsthm,epsfig}
\usepackage[usenames,dvipsnames]{xcolor}
\usepackage{bm,xspace}
\usepackage{tcolorbox}
\usepackage{cancel}
\usepackage{fullpage}
\usepackage{liyang}
\usepackage{framed}
\usepackage{verbatim}
\usepackage{enumitem}
\usepackage{array}
\usepackage{multirow}
\usepackage{afterpage}
\usepackage{mathrsfs}
\usepackage{pifont} 
\usepackage{chngpage}
\usepackage[normalem]{ulem}
\usepackage{boxedminipage}
\usepackage{caption}
\usepackage{thmtools}
\usepackage{thm-restate}

\DeclareMathOperator*{\argmin}{arg\,min}

\newcommand{\weight}{\mathrm{weight}}

\newcommand{\Geo}{\mathrm{Geo}}
\newcommand{\Uniform}{\mathrm{Uniform}}
\newcommand{\Ber}{\mathrm{Ber}}
\newcommand{\Poi}{\mathrm{Poi}}
\newcommand{\Categorical}{\mathrm{Categorical}}

\newcommand{\pj}{p^{(j)}}
\newcommand{\Dn}{\mathcal{D}^{(n)}}

\newcommand{\D}{\mathcal{D}}

\newcommand{\icx}{\preceq_{\mathrm{icx}}}
\newcommand{\cx}{\preceq_{\mathrm{cx}}}

\newcommand{\BS}{\boldsymbol{\mathcal{S}}}
\newcommand{\W}{\mathcal{W}}
\newcommand{\Rpos}{\mathbb{R}_{\geq 0}}
\newcommand{\Zpos}{\mathbb{Z}_{\geq 0}}

\usepackage[utf8]{inputenc}

\def\colorful{1}

\newcommand{\pparagraph}[1]{\bigskip \noindent {\bf {#1}}}

\ifnum\colorful=1

\fi
\ifnum\colorful=0

\fi

\title{Multiway Online Correlated Selection}
\date{\today}

\author{Guy Blanc \vspace{8pt}\\{\small {\sl Stanford University}}
\and \hspace{10pt} Moses Charikar \vspace{8pt} \\ \hspace{3pt} {\small {\sl Stanford University}}}  
\date{\today}

\begin{document}

\maketitle

\begin{abstract} 
We give a $0.5368$-competitive algorithm for edge-weighted online bipartite matching. Prior to our work, the best competitive ratio was $0.5086$ due to Fahrbach, Huang, Tao, and Zadimoghaddam (FOCS 2020). They achieved their breakthrough result by developing a subroutine called \emph{online correlated selection} (OCS) which takes as input a sequence of pairs and selects one item from each pair. Importantly, the selections the OCS makes are negatively correlated.

We achieve our result by defining \emph{multiway} OCSes which receive arbitrarily many elements at each step, rather than just two. In addition to better competitive ratios, our formulation allows for a simpler reduction from edge-weighted online bipartite matching to OCSes. 
While Fahrbach et al. used a factor-revealing linear program to optimize the competitive ratio, our analysis directly connects the competitive ratio to the parameters of the multiway OCS. Finally, we show that the formulation of Farhbach et al. can achieve a competitive ratio of at most $0.5239$, confirming that multiway OCSes are strictly more powerful.

\end{abstract} 

\thispagestyle{empty}
\newpage

\setcounter{page}{1}
\section{Introduction}

The matching problem has played a pivotal role in the development of algorithmic techniques in combinatorial optimization.
The online version of the problem was one of the first graph optimization problems studied in the online literature.
Karp, Vazirani and Vazirani \cite{KVV90} initiated a study of the problem and gave a $1-1/e$ competitive algorithm for unweighted graphs.
Here, the vertices of one side (the offline side) are known to the algorithm in advance.
The vertices of the other side appear in an online fashion.
As they appear, they must be matched irrevocably to an offline vertex (or discarded).
Since the work of \cite{KVV90}, online bipartite matching received considerable attention and their ideas were applied to a variety of other online assignment problems.

More than a decade later, the emergence of sponsored search auctions gave renewed impetus to the study of online matching and its variants.
There is a natural correspondence here -- the offline side corresponds to advertisers and the online side corresponds to impressions (i.e. opportunities to place ads).
Aggarwal et al \cite{AGKM11} generalized the \cite{KVV90} result to the vertex-weighted case, obtaining a $1-1/e$ competitive algorithm.

The edge-weighted case had remained a tantalizing open problem until very recently.
Here, it is not hard to see that no non-trivial competitive ratio is possible without giving the algorithm additional flexibility.
Researchers have focused on the {\em free disposal} setting where offline vertices can be assigned multiple online vertices but the contribution of an offline vertex to the objective function is the largest of the edge weights assigned to that vertex.
(i.e. previously matched edges can be disposed for free to make room for higher weight edges).
In the display advertising setting, this corresponds to assigning multiple impressions to an advertiser.

For the edge weighted problem, a natural greedy algorithm achieves competitive ratio 0.5.
This was the best known until the recent breakthrough by Fahrbach, Huang, Tao and Zadimoghaddam \cite{FHTZ20} who gave a 0.5086 competitive algorithm.
One of their main technical contributions was a novel algorithmic ingredient called {\em online correlated selection} (OCS).
This is an online subroutine that takes a sequence of pairs of vertices as input and selects one vertex from each pair such that decisions across pairs are negatively correlated.
A suitable quantification of this negative correlation was plugged into a factor revealing linear program to obtain a bound on the competitive ratio.

In this work we give generalize the OCS definition of Fahrbach et al \cite{FHTZ20} to multiway OCSes which take in multiple vertices at each step.
Our analysis directly connects the parameters of our multiway OCS to the eventual competitive ratio without the need to go through a factor revealing linear program.
We give a general framework for obtaining multiway OCSes and instantiate this framework
to obtain an improved competitive ratio of $0.5368$ for online weighted matching.
Finally, we show that the OCS definition of Fahrbach et al \cite{FHTZ20} cannot lead to a competitive ratio better than 0.5239 demonstrating that our multiway OCSes are strictly more powerful.

\subsection{Related Work}

Online weighted matching has been extensively studied in the literature. See the excellent survey by Mehta \cite{Meh13} for a comprehensive overview.
Two popular settings where competitive ratios better than 0.5 have been obtained are\\
\noindent{\bf (1)} the case where offline vertices have large capacity (i.e a large number of online vertices can be assigned to an offline vertex). 
This includes work of Kalyanasundaram and Pruhs \cite{KP00}, Feldman et al \cite{FKM09} on Display Ads and work on the Adwords problem \cite{MSVV05,BJN07}.\\
\noindent{\bf (2)} the case where stochastic information is available about the online vertices.
This includes work in the setting where online vertices are drawn from a known or unknown distribution \cite{FMMM09,KMT11,DJSW11,HMZ11,MGS12,MP12,JL13} or the setting that they arrive in a random order \cite{GM08,DH09,FHK10,MY11,MGZ12,MWZ15,HTWZ19}.\\
In addition to these, several recent advances have been made in more general settings including non-bipartite graphs and different arrival models \cite{HKT18,GKS19,GKM19,HPT19}.

\subsection{Problem definition: Edge-weighted online bipartite matching}

Consider some weighted bipartite graph $G=(L, R, w)$, where $L$ and $R$ are the left and right vertices respectively. If there is an edge between $i \in L$ and $j \in R$, then $w_{ij} > 0$ is the weight of that edge. Otherwise, $w_{ij} = 0$.

At the start, the algorithm is given the entire set of left vertices, $L$, but no information about $R$ or $w$. The vertices from $R$ appear in an online fashion, one by one. Hence, we refer to $L$ as the \emph{offline} vertices $R$ as the \emph{online} vertices. When an online vertex $j \in R$ appears, the entries $w_{ij}$ for each $i \in L$ are revealed to the algorithm. The algorithm must irrevocably decide which offline vertex to match $j$ to before the next online vertex appears.

The objective is to maximize the total weight of the matching. We operate in the \emph{free disposal model}. This means a single offline vertex $i$ may be matched to multiple online vertices, but only the weight of its heaviest edge is counted towards the objective.  We say a randomized algorithm is ``$\Gamma$-competitive" or ``has competitive ratio of $\Gamma$" if the expected objective of the algorithm's output is within a multiplicative factor of $\Gamma$ of the optimal objective with hindsight.

\section{Online correlated selection: Motivation and application}

\label{sec:OCS-motivation}

We withhold a full preliminaries section until \Cref{sec:preliminaries}, mentioning for now that we use {\bf boldface} (e.g.~$\bx\sim\D$) to denote random variables.

It is well-known that even for \emph{unweighted} online bipartite matching, deterministic algorithms cannot achieve a competitive ratio better than $0.5$. In this section, we'll investigate what form of randomness is needed to push the competitive ratio beyond $0.5$.

At the $j^{\text{th}}$ time step, the portion of the graph corresponding to $j \in R$ is revealed to the algorithm and it must irrevocably match it to some $\bi_j \in L$. In order to achieve a competitive ratio better than $0.5$, the choices $\bi_1, \bi_2, \ldots$ must use randomness. What if those choices are independent?

\begin{restatable}[]{lemma}{factindependent}
    \label{fact:independent-not-enough}
    Even for \emph{unweighted} online bipartite matching, no algorithm making decisions that are independent across time steps can achieve a competitive ratio better than $0.5$.
\end{restatable}
We include a proof of \Cref{fact:independent-not-enough} in \Cref{sec:proofs-OCS-motivation}. In order to break the $0.5$-competitive barrier, not only does the algorithm need to carefully select a marginal distribution for each $\bi_j$, it also needs some ``smarts" in deciding how to realize $\bi_j$ --- they cannot simply be independent. The key insight underlying Online Correlated Selection (OCS) is that these tasks can be tackled individually.

We separate the online algorithm into two components. The first component is fully deterministic and has access to the input instance. At the $j^{\text{th}}$ time step, it chooses a marginal distribution for $\bi_j$ as a deterministic function of the portion of the graph revealed so far. In particular, that marginal distribution does not depend on the realizations of $\bi_1, \ldots, \bi_{j-1}$.

The second component, the OCS, is given the marginal distribution $\bi_j$, but \emph{not} given any other information about the input instance --- It cannot ``directly see" the graph. At the $j^{\text{th}}$ time step, it chooses a realization for $\bi_j$ consistent with the desired marginal as a function of its own randomness. In particular, the realization for $\bi_j$ can depend on the previous realizations of $\bi_1, \ldots, \bi_{j-1}$. This separation of responsibilities is summarized in \Cref{fig:separation}.

\begin{figure}
    \begin{center}
         \begin{tabular}{ |  c || p{55mm} | p{55mm} | }
         \hline
          At time step $j \ldots$ & Deterministic Component & Randomized Component (OCS) \\ 
         \hline
         \hline
         Information revealed &The edge weights, $w_{ij}$ for each $i \in L$ &  The marginal distribution for $\bi_j$ (from the deterministic component) and the values for $\bi_j'$ for each $j' < j$  \\
         \hline
         Information hidden & The randomized selections of the OCS  & The input instance  \\
         \hline
         Output & A marginal distribution for $\bi_j$ & A realization for $\bi_j$ \\
      \hline
    \end{tabular}
    \end{center}
    \caption{We separate a matching algorithm into a deterministic component and randomized component, each with different responsibilities and access to information.}
    \label{fig:separation}
\end{figure}

In the following definition, a \emph{probability vector} is a nonnegative vector whose elements sum to $1$. It represents a marginal distribution over elements in $L$.

\begin{definition}[Continuous OCS]
    \label{def:continuous-OCS}
    Consider a set of ground elements, $L$. A \emph{continuous OCS} is an online algorithm that at each time step $j = 1,2,\ldots$, receives a probability vector $\pj$ with $|L|$ elements and irrevocably chooses a winner in $\bi_j \in L$ among those with $(\pj)_i > 0$.
\end{definition}

By \Cref{fact:independent-not-enough}, in order to achieve a competitive ratio of greater than $0.5$, the OCS cannot simply choose winners independently at each round according to the desired marginal distribution. In the next subsection, we examine what properties an OCS needs to have to allow for better competitive ratios.

\subsection{What makes a good OCS?}

The algorithm's goal is to maximize the expected weight of the matching over the randomness of the OCS. Having many online vertices $j \in R$ matched to the same offline vertex $i \in L$ can be wasteful. This is because only the heaviest match counts for the objective. To counteract this, the choices the OCS makes will be \emph{negatively correlated}. In particular we wish for, $j \neq j' \in R$, that:
\begin{align*}
    \Pr[\bi_j = \bi_{j'} = i] < (\pj)_i \cdot (p^{(j')})_i
\end{align*}
Unfortunately, it is impossible for the OCS to negatively correlated decisions across all time steps by any non-negligible amount. We formalize that impossibility in \Cref{lem:no-neg-always}, proved in \Cref{sec:proofs-OCS-motivation}. To get around that, we negatively correlated time steps that are ``close" temporally. Let $S = \{j_1, j_1 + 1, \ldots, j_2\}$ be a set of \emph{consecutive} time steps. We'll require our OCS to satisfy
\begin{align*}
    \Pr[\bi_j \neq i \text{ for every }j \in S] \leq f\left(\sum_{j \in S} \pj_i \right)
\end{align*}
where $f:\Rpos \to [0,1]$ is a function quantifying the negative correlation of an OCS. The OCS that chooses winners independently at each round according to the marginal distribution satisfies the above equation with $f_{\mathrm{trivial}}(x) \coloneqq e^{-x}$. OCSes with ``good" negative correlation have $f(x) < f_{\mathrm{trivial}}(x)$ and the amount $f$ is smaller than $f_{\mathrm{trivial}}$ directly connects to competitive ratio of the resulting algorithm.

The selections at time steps that are very far apart temporally will essentially be independent, as in the following definition.

\begin{definition}[Quantifying a continuous OCS, informal version of \Cref{def:OCS parameters}]
    \label{def:cont-OCS-parameters-informal}
    For any function $f: \R_{\geq 0} \to [0,1]$, an $f$-continuous OCS is a continuous OCS with the following guarantee. For any element $i \in L$ and $S_1, \ldots, S_k$ each sets of consecutive time steps, 
    \begin{align*}
        \Pr[\bi_j \neq i \text{ for every }j \in S_1 \cup \ldots \cup S_k] \leq \prod_{\ell = 1}^k f\left(\sum_{j \in S_k} (\pj)_i \right).
    \end{align*}
\end{definition}
We give an algorithm for edge-weighted online bipartite matching using an $f$-OCS as a subroutine and are able to directly connect the competitive ratio to $f$.
\begin{theorem}[Informal version of \Cref{thm:OCS-cont-Gamma}]
    \label{thm:OCS-cont-Gamma-informal}
    For any differentiable and convex $f: \Rpos \to [0,1]$, if there is an $f$-continuous OCS, there is an algorithm for edge-weighted online bipartite matching with competitive ratio of
    \begin{align*}
        \Gamma \coloneqq 1 - \int_{0}^\infty e^{-t} f(t)dt.
    \end{align*}
\end{theorem}
As expected, using $f_{\mathrm{trivial}}(x) \coloneqq e^{-x}$ gives a competitive ratio of $0.5$. Any $f$-OCS for $f$ strictly less than $f_{\mathrm{trivial}}$ gives a competitive ratio that breaks the $0.5$ barrier. 

The proof of \Cref{thm:OCS-cont-Gamma-informal} uses a similar online primal-dual formulation as that of \cite{FHTZ20}. We find that using our definition for an OCS in place of theirs substantially simplifies that analysis. In particular, we derive a simple analytic formula for the competitive ratio whereas they use a factor-revealing LP. The source of that simplification is that their algorithm needs to decide which of three round types, ``deterministic," ``randomized," or ``unmatched," to execute at each time step. Our continuous formulation implicitly interpolates between all three round types.

\subsection{Discrete OCSes}
We define discrete OCSes for two reasons.
\begin{enumerate}
    \item It allows us to relate our results to those of \cite{FHTZ20} who used (what we call) a discrete OCS.
    \item Later, to construct a continuous OCS, we will first construct a discrete OCS and use a reduction to turn it into a continuous OCS.
\end{enumerate}

The following definition of an $m$-discrete OCS is equivalent to a continuous OCS where each element in the probability vector is forced to be an integer multiple of $\frac{1}{m}$. Therefore, continuous OCSes are generalizations of discrete OCSes.

\begin{definition}[Discrete OCS]
    \label{def:discrete OCS}
    Consider a set of ground elements, $L$. An \emph{m-discrete OCS} is an online algorithm that at each time step $j = 1,2, \ldots$, is given a size-$m$ multiset $A_j$ of elements in $L$ and irrevocably selects a single winner $\bi_j \in A_j$.
\end{definition}
We quantify the quality of a discrete OCS similarly as for continuous OCSes.

\begin{definition}[Quantifying a discrete OCS]
    \label{def:discrete-parameters}
    For any $m \in \Z_{\geq 2}$ and $F: \Zpos \to [0,1]$, an $(F, m)$-discrete OCS is an $m$-discrete OCS with the following guarantee. For any element $i \in L$ and $S_1, \ldots, S_k$ each sets of consecutive time steps,
    \begin{align*}
        \Pr[\bi_j \neq i \text{ for every }j \in S_1 \cup \ldots \cup S_k] \leq \prod_{\ell = 1}^k F\left(\sum_{j \in S_k} \mathrm{count}_{i}(A_j) \right).
    \end{align*}
    where $\mathrm{count}_{i}(A_j)$ is the number of times $i$ appears in the multiset $A_j$.
\end{definition}

For example, the OCS that independently and uniformly $\bi_j$ from  $A_j$ at each round achieves $F_{\mathrm{trivial}}(n) = (1 - 1/m)^n$. Our goal is to construct OCSes with $F < F_{\mathrm{trivial}}$, which directly leads to competitive ratios better than $\frac{1}{2}$.

\begin{theorem}[Informal version of \Cref{thm:discrete-competitive-ratio-formal}]
    \label{thm:discrete-competitive-ratio-informal}
    For any $m \in \N$ and convex $F: \Z_{\geq 0} \to [0,1]$, if is there an $(F,m)$-discrete OCS, there is an algorithm for edge-weighted online bipartite matching with competitive ratio of
    \begin{align*}
        \Gamma \coloneqq 1 - \sum_{n = 0}^\infty \frac{m^n}{(m+1)^{n+1}} \cdot F(n).
    \end{align*}
\end{theorem}
We prove (a formal version of) \Cref{thm:discrete-competitive-ratio-informal} by showing how to construct a continuous OCS from a discrete OCS and then applying \Cref{thm:OCS-cont-Gamma-informal}. This is done in \Cref{sec:cont-to-disc-OCS}.

Lastly, we connect our OCS definition to that of Fahrbach et al.

\begin{definition}[$\gamma$-OCS, Definition 2 of \cite{FHTZ20}]
    \label{def:OCS-fahrbach}
    For any $\gamma \in (0,1)$, a $\gamma$-OCS is an $(F,2)$-discrete OCS where 
    \begin{align*}
        F(n) \coloneqq 2^{-n} \cdot (1 - \gamma)^{\max(n-1, 0)}.
    \end{align*}
\end{definition}
As a straightforward application of \Cref{thm:discrete-competitive-ratio-informal}, we derive an explicit formula mapping $\gamma$ to a competitive ratio.

\begin{restatable}[]{theorem}{corfahr}
    \label{cor:Gamma-fahrbach}
    For any $\gamma \in (0,1)$, if there is a $\gamma$-OCS, there is an algorithm for edge-weighted online bipartite matching with a competitive ratio of
    \begin{align}
        \label{eq:Gamma-fahrbach}
        \Gamma = \frac{3 + 2 \gamma}{6 + 3\gamma}.
    \end{align}
\end{restatable}
We experimentally compared the competitive ratio from \Cref{eq:Gamma-fahrbach} to those of the gain-sharing LP used in \cite{FHTZ20}. For every $\gamma \in \{0, 0.01, \ldots 0.46\}$, the competitive ratios are within an additive $10^{-8}$ of one-another. That small difference is likely due to rounding issues. For $\gamma \in \{0.47, 0.48, \ldots, 1\}$, our equation gives a greater competitive ratio. Note that, as we later prove in \Cref{lem:negative}, $\gamma$-OCSes can only exist for $\gamma \leq \frac{1}{3}$ and our competitive ratio appears to match theirs on that range.

Lastly, we remark that ``perfect negative correlation" for \Cref{def:OCS-fahrbach} corresponds to $\gamma = 1$. As noted in \cite{FHTZ20} and can be confirmed by plugging $\gamma = 1$ into \Cref{cor:Gamma-fahrbach}, that corresponds to a competitive ratio of $\Gamma = \frac{5}{9}$. This is far from the best known upper bound of $\Gamma = 1 - \frac{1}{e}$. On the other hand, perfect negative correlation for a continuous OCS corresponds to $f(x) = (1 - x)_+$ which gives a competitive ratio of $\Gamma = 1 - \frac{1}{e}$. While perfect negative correlation is impossible, the fact that it corresponds to a competitive ratio equal to the best known upper bound seems to suggest that continuous OCSes are a step in the right direction towards finding the optimal competitive ratio.

\subsection{Negative results}

We prove the following for Fahrbach et al.'s formulation of an OCS.
\begin{restatable}{lemma}{negative}
    \label{lem:negative}
    No $\gamma$-OCS exists for $\gamma > \frac{1}{3}$.
\end{restatable}
    
Both \Cref{cor:Gamma-fahrbach} and the factor-revealing LP of \cite{FHTZ20} agree that the competitive ratio for $\gamma = \frac{1}{3}$ is less than $0.5239$. This is an upper bound on the competitive ratios that can be achieved simply by constructing better $\gamma$-OCSes. Our competitive ratio of $0.5368$ is larger, confirming that multiway OCSes are strictly more powerful.

\section{Constructing a multiway OCS: Meta algorithm overview}
\label{subsec:construction-overview}

We give a meta algorithm for constructing $m$-discrete OCSes given a ``win distribution," $\W$, over $[0,1]^\N$. A sample from $\mathcal{W}$ is a sequence of probabilities, $(\bx_1, \bx_2, \ldots)$, with the quality of the resulting OCS (i.e. the function $F$ in \Cref{def:discrete-parameters}) depending on how negatively correlated the elements in that sequence are. During initialization, the meta algorithm draws an independent ``win sequence", $(\bx_1(i), \bx_2(i), \ldots)$ for each $i \in L$. Our OCS also stores an index within the win distribution, $k(i)$ for each $i \in L$ that is initialized to $k_i = 1$.

At each time step, the OCS receives $A_j = \{i_1, \ldots, i_m \}$. For each $i \in A_j$, it determines the ``desired win probability" of $i$ in this time step. In the case where $i$ appears exactly once in $A_j$, this desired win probability is exactly equal to $\bx_{k(i)}(i)$. When $i$ appears $r \geq 2$ times in $A_j$, the desired win probability is the probability $i$ would win in at least one of the next $r$ rounds if it appeared exactly once in each. Then, a tournament algorithm takes in these desired win probabilities and selects a single winner for the round consistent with the desired win probabilities and $k(i)$ is incremented $r$ times. We include pseudocode for our meta algorithm in \Cref{fig:OCS meta algorithm}, though defer description of the win distribution and tournament subroutine for later.

\begin{figure}[h]
  \captionsetup{width=.9\linewidth}
\begin{tcolorbox}[colback = white,arc=1mm, boxrule=0.25mm]
\vspace{3pt} 
$\mathrm{OCS}_{\mathcal{W}, \mathcal{T}}$
\begin{enumerate}
    \item \textbf{Initialization:} For each $i \in L$, independently draw a $\bx(i) \sim \mathcal{W}$ and initialize an index $k_1(i) \leftarrow 1$
    \item Upon receiving $A_j = \{i_1, \ldots, i_m\}$
    \begin{enumerate}
        \item \textbf{Determine desired win probabilities:} For each $i \in A_j$, let $r_j(i)$ be it's multiplicity in $A_j$. Compute the desired win probability
        \begin{align}
            \label{eq:desired-win-prob-prod}
            \bw_j(i) \coloneqq 1 - \prod_{\ell = 0}^{r_j(i)-1} \left(1 - \bx_{k_j(i) + \ell}\right)
        \end{align}
        \item \textbf{Choose a winner:} Use a (possibly randomized) tournament subroutine, $\mathcal{T}$, to select a single $i \in A_j$ as winner. This is done so that the probability each $i \in A_j$ wins is at least $\bw_j(i)$.
        \item \textbf{Increment counters:} For each $i \in L$, increment $k_{j+1}(i) \leftarrow k_j(i) + r_j(i)$
    \end{enumerate}
\end{enumerate}
\end{tcolorbox}
\caption{Our meta algorithm for constructing OCSes as a function of a win distribution, $\mathcal{W}$ and tournament subroutine $\mathcal{T}$}
\label{fig:OCS meta algorithm}
\end{figure}

The fact that the desired win probabilities of different vertices are \emph{independent} makes analysis tractable. To achieve a competitive ratio better than $0.5$, we need the winners of our OCS to be negatively correlated. We achieve this by making the random variables in the sequence returned by the win distribution negatively correlated. This way, if $i$ has a small chance of winning in one round, it has a larger chance of winning in other rounds. In particular, we need two types of negative correlation for $(\bx_1, \bx_2, \ldots) \sim \W$.

\begin{enumerate}
    \item A \emph{qualitative} and \emph{long-distanced} negative correlation. We need to ensure there doesn't exist any positive correlation between $\bx_{k_1}$ and $\bx_{k_2}$ even for $k_1$ and $k_2$ far apart. We prove that our win probabilities satisfy a generalization of negatively correlated random variables called \emph{negative association}. 
    \begin{restatable}[Negative association \cite{JP83}]{definition}{NA}
        \label{def:NA}
        Random variables $\bx_1, \bx_2, \ldots, \bx_n$ are said to be \emph{negatively associated} if, for every pair of disjoint subsets $A_1, A_2 \subseteq [n]$ and nondecreasing functions $f_1, f_2$,
            \begin{align*}
                \Cov[f_1(\bx_{A_1}), f_2(\bx_{A_2})] \leq 0
            \end{align*}
            where $\bx_{A}$ is the vector $(\bx_{i_1}, \bx_{i_2}, \ldots, \bx_{i_{|A|}})$ when $A = \{i_1, i_2, \ldots, i_{|A|}\}$.
    \end{restatable}
Proving this ends up being essential to showing our OCS has \emph{any} guarantees, even those matching the trivial fully independent OCS.
    \item A \emph{quantitative} and \emph{short-distanced} negative correlation. We want the values of $\bx_{j}$, and $\bx_{j'}$ to be strongly negatively correlated whenever $j$ is ``close to" $j'$. The magnitude of this negative correlation directly ties into the probability that $i$ is selected in consecutive time steps and therefore the quality of our OCS.
\end{enumerate}

We prove a meta theorem on the quality of the algorithm in \Cref{fig:OCS meta algorithm} as a function of the win distribution. 
\begin{restatable}[Quality of our meta algorithm]{theorem}{metaalg}
    \label{thm:meta-algorithm}
    Let $\mathcal{W}$ be a win-distribution with the following three properties.
    \begin{enumerate}
        \item \textbf{Negatively associated:} For any $n \geq 1$ and $\bx \sim \W$, the variables $\bx_1, \ldots, \bx_n$ are negative associated.
        \item \textbf{Shift-invariant:} For any $n \geq 1$ and $\bx \sim \W$, the distribution of $(\bx_1, \bx_2, \ldots)$ is identical to that of $(\bx_n, \bx_{n+1}, \ldots)$.
        \item \textbf{Sufficiently small:} For each $r \in [m-1]$ and $t \in [0,1]$, defining $\bw^{(r)}$ to be the random variable representing the distribution of the desired win probability for a vertex of multiplicity $r$,
    \begin{align*}
       \bw^{(r)} \sim \left(1 - \prod_{\ell = 1}^{r} \left(1 - \bx_{\ell}\right) \right) \quad\quad \text{where} \quad\quad \bx \sim \W,
    \end{align*} the following holds
        \begin{align}   
            \label{eq:w-not-too-large}
            \Ex_{\bx \sim \W} \left[\left(\bw^{(r)} - t \right)_+ \right] \leq 1 - t  - \frac{m-r}{m} \cdot \left (1 - t^{\frac{m}{m-r}} \right)
        \end{align}
        where $(z)_+$ is short hand for $\max(z, 0)$.
    \end{enumerate}
    Then, there exists some tournament subroutine $\mathcal{T}$ for which $\mathrm{OCS}_{\mathcal{W}, \mathcal{T}}$ is a $(f,m)$-discrete OCS for
    \begin{align}
        \label{eq:def f}
        F(n) \coloneqq \Ex_{\bx \sim \W} \left[\prod_{\ell = 1}^n(1 - \bx_\ell) \right].
    \end{align}
\end{restatable}
Each of the three criteria of \Cref{thm:meta-algorithm} are desirable for intuitive reasons. Let $S_1, \ldots, S_k$ be disjoints set of time steps. Negative association ensures that,
\begin{align*}
    \Pr[i \text{ never selected in } \bigcup_{\ell = 1}^k S_\ell] \leq \prod_{\ell  = 1}^k \Pr[i \text{ never selected in }S_\ell].
\end{align*}
Shift-invariance ensures that the probability that $i$ is not selected in a set of $n$ consecutive appearances is independent of when those appearances occur. Finally, $\bw^{(r)}$ being sufficiently small ensures that it's possible for a tournament subroutine to guarantee that the probability $i$ wins a round with desired win probability $w$ is at least $w$.

Our proof of \Cref{thm:meta-algorithm} leverages a number of tools from probability theory. We use many closure properties of negatively associated variables given in \cite{JP83} and also prove some additional ones. Later, we formalize what it means for one variable to be ``larger and more spread out" than another using the \emph{increasing convex order}. We do so to take advantage of a powerful result by  M{\"u}ller and R{\"u}schendorf \cite{MR01} which generalizes a classic result of Strassen \cite{s65}. That result guarantees the existence of a particular coupling $(\ba, \bb)$ whenever $\bb$ is ``larger and more spread out" than $\ba$.

\section{Organization of the remainder of this paper}

\Cref{sec:concurrent} compares our results to concurrent works and \Cref{sec:preliminaries} summarizes our notation. The remainder of this paper is devoted to proving our main theorem:
\begin{restatable}[]{theorem}{main-thm}
    \label{thm:main}
    There is a $\Gamma = 0.5368$-competitive algorithm for edge-weighted online bipartite matching.
\end{restatable}
The sections are ordered as successive reductions from edge-weighted online bipartite matching with some proofs deferred to the appendix.
\begin{itemize}
    \item In \Cref{sec:continuous-OCS-formal}, we give the formal version of \Cref{def:cont-OCS-parameters-informal}, specifying the requirements a continuous OCS has to meet to be useful for edge-weighted online bipartite matching. We also state and prove the formal version of \Cref{thm:OCS-cont-Gamma-informal} mapping that definition to a competitive ratio.
    \item In \Cref{sec:cont-to-disc-OCS}, we show to construct a continuous OCS from a discrete OCS and then give the formal version of \Cref{thm:discrete-competitive-ratio-informal} mapping a discrete OCS to a competitive ratio. We also give the proof of \Cref{cor:Gamma-fahrbach} and of \Cref{thm:main} as a consequence of the discrete OCS construction in later sections.
    \item In \Cref{sec:metathm-proof}, we prove \Cref{thm:meta-algorithm} which gives a discrete OCS from a win distribution.
    \item In \Cref{sec:W}, we construct a win distribution meeting the requirements of \Cref{thm:meta-algorithm} and optimize its parameters to achieve a $0.5368$-competitive algorithm.
    \item Finally, in  \Cref{sec:neg}, we provide a short proof of \Cref{lem:negative} upper bounding the competitive ratio achievable by the OCS definition of \cite{FHTZ20}. 
\end{itemize}

\section{Comparison to the concurrent work}
\label{sec:concurrent}

\subsection{The work of Gao et al.}
Concurrently and independently, Gao et al.\ also explored the use of multiway OCSes in the construction of improved online matching algorithms \cite{GHHNYZ21}. The bulk of their work focuses on \emph{vertex-weighted} online bipartite matching. This is a special case of \emph{edge-weighted} online bipartite matching where we assume that all edges to each offline vertex have the same weight. While an algorithm achieving the optimal competitive ratio competitive ratio of $(1 - \frac{1}{e}) \approx 0.632$ is known in this setting \cite{AGKM11}, Gao et al.\ are able to show that OCSes can also be used to tackle this problem and achieve a competitive ratio of $0.593$.

They do so by developing a definition of multiway OCSes similar to our definition of continuous OCSes. In the vertex-weighted setting we only care about whether an element $i$ is selected at least once at some time step and do not care about which time step it is selected in. Their definition for a ``multi-way semi-OCS" is equivalent to \Cref{def:cont-OCS-parameters-informal} where we take $k = 1$ and $S_1$ the set containing all time steps. For vertex-weighted online bipartite matching, they derive the same formula for the competitive ratio as we do in \Cref{thm:OCS-cont-Gamma-informal} using ``multi-way semi-OCSes." A similar result could be derived using our algorithm by never executing \Cref{eq:reset} (since we only care about whether a vertex is matched, not when it is matched). The resulting algorithm and analysis simplifies.

We stress that extending multiway OCSes to the historically challenging edge-weighted setting requires overcoming additional hurdles.
\begin{enumerate}
    \item The informal definition of \Cref{def:cont-OCS-parameters-informal} isn't strong enough to derive \Cref{thm:OCS-cont-Gamma-informal}. We needed to identify an additional constraint on the OCS (see ``partially ignore small gaps" of \Cref{def:OCS parameters}) that is both feasible to construct and sufficient to developing an algorithm for edge-weighted online bipartite matching.
    \item Constructing a multiway OCS is more challenging than a ``multi-way semi-OCS" and we need to develop a new framework to do so. For example, in the construction of Gao et al., the probability a vertex is selected in a late time step can be made arbitrarily small (since they only need each vertex to be matched once), and so would not meet our definition.
\end{enumerate}

Gao et al.\ also contribute to the \emph{edge-weighted} setting. In this setting, they give $0.519$-competitive algorithm by constructing an improved $2$-way OCS. Their result improves upon the competitive ratio of Fahrbach et al. \cite{FHTZ20} though isn't as large as ours. Finally, they prove that no $\gamma$-OCS exists for $\gamma > \frac{1}{4}$, which is stronger than the $\frac{1}{3}$ bound we show in \Cref{lem:negative}. For our purposes, the conclusion is the same: Multiway OCSes are strictly more powerful than two-way OCSes for edge-weighted online bipartite matching.

\subsection{The work of Shin and An}

Also concurrently and independently, Shin and An created a $3$-discrete OCS an applied it to edge-weighted online bipartite matching \cite{SA21}. They construct their $3$-discrete OCS by running two parallel instance of a $\gamma$-OCS (which is $2$-discrete). Upon receiving a set of three elements, two are chosen uniformly at random to be passed to one of the $\gamma$-OCSes. The winner and the third element are passed into the other $\gamma$-OCS to decide the overall winner of the round.

Then, they show how to use this $3$-discrete OCS for edge-weighted online bipartite matching. To do so, they extend the analysis and factor-revealing LP of Farhbach et al. and derive a competitive ratio of $0.5093$. We note that plugging in the parameters of their $3$-discrete OCS into \Cref{thm:discrete-competitive-ratio-informal} gives the same competitive ratio up to the four digits of accuracy provided in their paper, though with a completely different derivation than theirs.

\section{Preliminaries}
\label{sec:preliminaries}

For brevity, we use \emph{decreasing} in place of \emph{nondecreasing} and \emph{increasing} in place of \emph{nondecreasing}. The notation $(x)_+$ is used as shorthand for $\max(x, 0)$. 

We use {\bf boldface} (e.g.~$\bx\sim\D$) to denote random variables. The notation $\ba \overset{d}{=} \bb$ signifies that $\ba$ and $\bb$ have the same distribution. The following distributions are of interest.
\begin{enumerate}
    \item The uniform distribution $\bx \sim \Uniform(0,1)$, gives a continuous $\bx$ equally likely to be anywhere on the interval $[0,1]$.
    \item The Bernoulli distribution, $\bx \sim \Ber(p)$, gives a discrete $\bx$ that is $1$ with probability $p$ and $0$ otherwise.
    \item The geometric distribution, $\bx \sim \Geo(p)$ gives $i - 1$, where $i$ is the smallest index such that $\by_i = 1$ for $\by_1, \by_2, \ldots \overset{\mathrm{iid}}{\sim}\Ber(p)$.
    \item The Poisson distribution, $\bx \sim \Poi(\lambda)$ gives $\bx$ defined as the following limit
    \begin{align*}
        \bx \coloneqq \lim_{n \to \infty} \left(\sum_{i = 1}^n \bx_i \quad \text{where} \quad \bx_i \overset{\mathrm{iid}}{\sim} \Ber(\lambda/n)\right).
    \end{align*}
    \item The categorical distribution, $\bi \sim \Categorical(p_1, \ldots, p_m)$ where $\sum_{j = 1}^m p_j = 1$, gives a discrete $\bi$ where $\bi = i$ with probability $p_i$.
\end{enumerate}

We'll need various ways to compare distributions. The simplest is \emph{stochastic dominance}.
\begin{definition}[Stochastic dominance]
    \label{def:stochastic-dominance}
    Let $\ba$ and $\bb$ be random variables over $\R$. We say that $\ba$ \emph{stochastically dominates} $\bb$, denoted $\ba \succeq \bb$ if, for any $t \in \R$,
    \begin{align*}
        \Pr[\ba \leq t] \leq \Pr[\bb \leq t]
    \end{align*}
\end{definition}

\begin{definition}[Stochastically increasing/decreasing functions]
    \label{def:stochastic-increasing}
    A stochastic function $\boldf$ from $\R^n$ to a distribution over $\R$ is said to be \emph{increasing} in coordinate $i \in [n]$ if for any $x,y \in \R^n$ where $x_j = y_j$ for each $j \neq i$ and $x_i \geq y_i$, $\boldf(x)$ stochastically dominates $\boldf(y)$. It is said to be stochastically \emph{decreasing} in coordinate $i$ if the $\boldf(y)$ stochastically dominates $\boldf(x)$ under the same conditions. A function is stochastically increasing (or decreasing) if it is stochastically increasing (or decreasing) in all of its coordinates.
\end{definition}

In some proofs, we'll need for one distribution to be ``less spread out" than another. We use the following two notions. For a more thorough overview, see \cite{SS07}.
\begin{definition}[Convex order]
    \label{def:cx}
    Let $\ba,\bb$ be two univariate random variables such that,
    \begin{align*}
        \Ex[\phi(\ba)] \leq \Ex[\phi(\bb)]
    \end{align*}
    for all convex functions $\phi: \R \to \R$ for which the above expectations exist. Then, we say $\ba$ is smaller than $\bb$ in the \emph{convex order} and denote this,
    \begin{align*}
        \ba \cx \bb.
    \end{align*}
\end{definition}

Similar to the above, but where we restrict $\phi$ to be increasing,

\begin{restatable}[Increasing convex order]{definition}{icxdef}
    \label{def:icx}
    Let $\ba,\bb$ be two univariate random variables such that,
    \begin{align*}
        \Ex[\phi(\ba)] \leq \Ex[\phi(\bb)]
    \end{align*}
    for all increasing convex functions $\phi: \R \to \R$ for which the above expectations exist. Then, we say $\ba$ is smaller than $\bb$ in the \emph{increasing convex order} and denote this,
    \begin{align*}
        \ba \icx \bb.
    \end{align*}
\end{restatable}
Finally, we'll use \emph{discrete derivatives}.
\begin{definition}[Discrete derivative]
    \label{def:disc-derivative}
    For any function $F: \Z \to \R$, the discrete derivative of $F$ is defined as
    \begin{align*}
        (\Delta F)(n) \coloneqq F(n+1) - F(n).
    \end{align*}
    We use the notation $\Delta^{(\ell)}$ to denote the discrete derivative operator composed $\ell$ times.
\end{definition}

\section{Continuous OCS: Formal definition and theorems}
\label{sec:continuous-OCS-formal}

In this section we formalize our definition of the parameters of a continuous OCS and the relation between those parameters and the competitive ratio. Recall in \Cref{def:cont-OCS-parameters-informal}, we guarantee the probability $i$ isn't picked in two consecutive sequences of time steps, $S_1$ and $S_2$, is at most
\begin{align}
    \label{eq:prob-separate}
    f(\weight_i(S_1)) \cdot f(\weight_i(S_2)).
\end{align}
where $\weight_i(S) = \sum_{j \in S} (\pj)_i$. On the other hand, if $S_1$ and $S_2$ are actually consecutive, then the probability $i$ isn't picked is at most
\begin{align}
    \label{eq:prob-combined}
    f\big(\weight_i(S_1) + \weight_i(S_2)\big).
\end{align}
Good OCSes will have $f(x + y) < f(x) \cdot f(y)$ and so \Cref{eq:prob-combined} gives a stronger guarantee than \Cref{eq:prob-separate}. Now, suppose that $S_1$ and $S_2$ are ``almost" consecutive. For $S_1 = \{j_1, \ldots, j_2\}$ and $S_2 = \{j_3, \ldots, j_4\}$, perhaps its the case that $j_3 = j_2 + 2$. Then, intuitively, we should expect a guarantee close to that of \Cref{eq:prob-combined} to hold. 

Let $G$ be the time steps in the ``gap", meaning $G \coloneqq \{j_2 + 1, \ldots, j_3 - 1\}$. For some rate parameter $r$ and $w \coloneqq \prod_{j \in G}(1- \cdot \pj_i)^r$, we will guarantee that the probability $i$ is not selected in $S_1$ or $S_2$ is at most
\begin{align*}
    w \cdot f\big(\weight_i(S_1) + \weight_i(S_2)\big) + (1 - w) \cdot f(\weight_i(S_1)) \cdot f(\weight_i(S_2))
\end{align*}
Essentially, when determining the probability $i$ is picked, we pretend the gap $G$ doesn't exist with probability $w$. This guarantee smoothly interpolates between \Cref{eq:prob-combined,eq:prob-separate} as $S_1$ and $S_2$ grow further apart. We generalize this idea to more than two sets in the following definition.

\begin{definition}[Parameters of a continuous OCS, formal version of \Cref{def:cont-OCS-parameters-informal}]
    \label{def:OCS parameters}
    For any rate $r \geq 0$ and function $f: \R_{\geq 0} \to [0,1]$, an $(f, r)$-\emph{continuous OCS} is a continuous OCS with the following guarantee. For any element $i \in L$ and set of time steps $S$, 
    \begin{align*}
        \Pr[\bi_j \neq i \text{ for every }j \in S] \leq \mathcal{F}(S, p)
    \end{align*}
    where $p$ is the vector whose $j^{\mathrm{th}}$ entry is $(\pj)_i$ and $\mathcal{F}$ satisfies
    \begin{enumerate}
        \item \textbf{Consecutive:} If $S = \{j_1, \ldots, j_2\}$, then
        \begin{align*}
            \mathcal{F}(S, p) \leq f\left(\sum_{j \in S} p_j\right)
        \end{align*}
        \item \textbf{Partially ignore small gaps:} For any disjoint (not necessarily consecutive) $S_1, S_2 \subseteq \N$ satisfying
        \begin{align*}
            j_1 < j_2 \text{ for each $j_1 \in S_1, j_2 \in S_2$}
        \end{align*}
        Let $j^\star$ be some time step between $j_1, j_2$ for every $j_1 \in S_1, j_2 \in S_2$. For
        \begin{align*}
            (\mathrm{del}_{j^\star}(p))_{j} &\coloneqq \begin{cases}
                p_{j} & j < j^\star\\
                p_{j + 1} & j \geq j^\star
            \end{cases} \\
            \mathrm{del}_{j^\star}(S_1 \cup S_2) &\coloneqq \{j \text{ for each }j \in S_1\} \cup \{j-1 \text{ for each }j \in S_2\}
        \end{align*}
        Then, for $w = (1 - p_{j^\star})^r$,
        \begin{align}
            \label{eq:ignore-gaps}
            \mathcal{F}(S_1 \cup S_2, p) \leq w \cdot \mathcal{F}(\mathrm{del}_{j^\star}(S_1 \cup S_2), \mathrm{del}_{j^\star}(p)) + (1 - w) \cdot \mathcal{F}(S_1, p) \cdot \mathcal{F}(S_2, p)
        \end{align}
    \end{enumerate}
\end{definition}

Smaller $r$ corresponds to ignoring larger gaps, making the OCS anticorrelated over larger distances. Instead of setting $w = (1 - p_{j^\star})^r$, we could have instead set $w = e^{-p_{j^\star}\cdot r}$, $w = (1 - r \cdot p_{j^\star})_+$, or other functions where $-r \leq \frac{dw}{dp_{j^\star}}\leq 0$, and our proof of \Cref{thm:OCS-cont-Gamma} would still work. The choice we make here plays well with the proof of \Cref{lem:disc-to-cont}. 

We next give a formal version of \Cref{thm:OCS-cont-Gamma-informal}, with the main difference being that $r$ can't be too large.

\begin{theorem}[Formal version of \Cref{thm:OCS-cont-Gamma-informal}]
    \label{thm:OCS-cont-Gamma}
    For any thrice differentiable convex $f:\R_{\geq 0} \to [0,1]$ where $f$ and $f'$ are log-concave, meaning
    \begin{align*}
        f''(x) \cdot f(x) \leq (f'(x))^2 \quad \quad \text{for all $x \geq 0$} \\
       f'''(x) \cdot f'(x) \leq (f''(x))^2 \quad \quad \text{for all $x \geq 0$}
    \end{align*}
    let,
    \begin{align}
        \label{eq:def-Gamma}
        \Gamma \coloneqq 1 - \int_{0}^\infty e^{-t} f(t)dt.
    \end{align}
    For any $r$ satisfying
    \begin{align}
        \label{eq:r-upper-bound-thm}
        r \leq \frac{-f'(0) - \Gamma}{\Gamma + (1 - \Gamma) \cdot f'(0)}
    \end{align}
    If there is a $(f,r)$-continuous OCS, then there is a $\Gamma$-competitive algorithm for edge-weighted online bipartite matching.
\end{theorem}
In most cases, $f'(0) = -1$ as OCSes of interest have $f$ between $f_{\mathrm{trivial}}= e^{-x}$ and $f_{\mathrm{perfect}} = (1 - x)_+$. When $f'(0) = 1$, \Cref{eq:r-upper-bound-thm} imposes an upper bound on $r$ of $\frac{1 - \Gamma}{2\Gamma - 1}$, which is a decreasing function of $\Gamma$ on the interval $(0.5, 1)$. Hence, the larger a competitive ratio we want, the smaller $r$ has to be, meaning the OCS must be anticorrelated over larger distances.

The remainder of this section is devoted to proving \Cref{thm:OCS-cont-Gamma}. The algorithm uses an online primal-dual approach.

\subsection{Primal dual formulation}
\label{sec:primal-dual}
 We begin with the standard primal linear program for (offline) edge-weighted bipartite matching. Let $w_{ij}$ be the weight of the edge between $i \in L$ to $j \in R$ where $w_{ij} = 0$ indicate no such edge exists and $x_{ij}$ be the probability that $(i,j)$ is the heaviest edge matched to vertex $i$.
\begin{align*}
    \max \quad &\sum_{i \in L} \sum_{j \in R} w_{ij} x_{ij} & \\
    \mathrm{s.t.} \quad &\sum_{j \in R} x_{ij} \leq 1 & \forall i \in L \\
    &\sum_{i \in L} x_{ij} \leq 1 & \forall j \in R \\
    &x_{ij} \geq 0 & \forall i \in L, j \in R
\end{align*}
This gives rise to a dual linear program.
\begin{align*}
    \min \quad & \sum_{i \in L} \alpha_i + \sum_{j \in R} \beta_j & \\
    \mathrm{s.t.} \quad &\alpha_i + \beta_j \geq w_{ij} & \forall i \in L, j \in R \\
    &\alpha_i \geq 0 & \forall i \in L \\
    & \beta_j \geq 0 &\forall j \in R
\end{align*}
We use $P$ to denote the primal objective and $D$ to denote the dual objective. Our online algorithm will simultaneously maintain a matching for the graph and a dual assignment that will meet the following two criteria.
\begin{fact}[Lemma 2 from \cite{FHTZ20}]
    Suppose an algorithm maintains primal and dual assignments such that, for some $0 \leq \Gamma \leq 1$, the following conditions hold at all time steps.
    \begin{enumerate}
        \item \textbf{\emph{Approximate dual feasibility:}} For any $i \in L$ and $j \in R$
        \begin{align*}
            \alpha_i + \beta_j \geq \Gamma \cdot w_{ij}
        \end{align*}
        
        \item \textbf{\emph{Reverse weak duality:}} The objectives of the primal and dual assignments satisfy $P \geq D$.
    \end{enumerate}
    Then the algorithm is $\Gamma$-competitive.
\end{fact}

Inspired by \cite{DHKMY16,FHTZ20}, we do not directly maintain the primal variables $x_{ij}$. Instead, for each $i \in L$ and weight-level $w \geq 0$, we define
\begin{align*}
    y_i(w) \coloneqq \Pr[i \text{ not matched to a vertex of weight at least }w]
\end{align*}
In contrast, \cite{DHKMY16,FHTZ20} defines their $\overline{y_i(w)}$ to be $\overline{y_i(w)} = 1 - y_i(w)$. We find our ``reversed" definition simplifies later analysis. The primal objective is then
\begin{align*}
    P = \sum_{i \in L} \int_{0}^{\infty} (1 - y_i(w))dw
\end{align*}
Similarly, we introduce and maintain new variables $\alpha_i(w)$ for each $i \in L$ and weight-level $w \geq 0$ and then set the value of $\alpha_i$ to
\begin{align*}
    \alpha_i = \int_{0}^\infty (\Gamma - \alpha_i(w))dw.
\end{align*}

\begin{remark}
    \label{remark:step function}
    Note that $y_i(w)$ and $\alpha_i(w)$ are both
    step functions with at most one change per unique value in the set \{$w_{ij} \mid j \in R\}$. Hence, they can be stored efficiently.
\end{remark}

\subsection{High level algorithmic overview}

Fix some $(f, r)$-continuous OCS. At each time step $j \in [T]$, the values $w_{ij}$ are revealed for every $i \in L$. As a function of these weights and the state of the primal and dual variables, the algorithm will determine a probability vector $\pj$ to pass into the OCS. Then, it updates the primal variables to account for the probability each vertex is matched to $j$ and the dual variables to maintain the conditions dual feasibility and reverse weak duality. The OCS selects a winner in $L$ based on the $\pj$ it is passed in, and that winner is matched to vertex $j$. We give pseudocode for the entire algorithm at the end of this section in \Cref{fig:MatchOnline}.

We wish to emphasize that the primal-dual algorithm does not ``know" which offline vertex $j$ is matched too. Instead, it only maintains the probability that each $i$ is matched to a vertex at each weight-level. This separates the algorithm into a deterministic component, the primal-dual portion, and a randomized component, the OCS.

The remainder of this section is divided as follows. In \Cref{subsec:primal updates} and \Cref{subsec:dual updates}, we define how the primal and dual variables respectively are updated given $\pj$. In particular, $\beta_j$ is set to exactly the minimum value so that approximate dual feasiblity holds. In \Cref{subsec:determining p} we will discuss how $\pj$ is determined. In \Cref{subsec:reverse weak duality}, we will prove that reverse weak duality holds, completing the proof that our algorithm is $\Gamma$-competitive.

Throughout this section, $f: \R_{\geq 0} \to [0,1]$ and $r \in \R^+$ refer to the parameters of our OCS.

\subsection{Primal updates}
\label{subsec:primal updates}

For each $i \in L$ and weight-level $w \geq 0$, we maintain two random variables initialized as follows.
\begin{align*}
    \bx_i(w) \overset{\mathrm{init}}{\longleftarrow} 0 \quad \text{with probability $1$}\\
    \bc_i(w) \overset{\mathrm{init}}{\longleftarrow} 1 \quad \text{with probability $1$}
\end{align*}

Our algorithm maintains the \emph{distribution} over $(\bx_i(w), \bc_i(w))$ rather than just a single element of it. See \Cref{remark:dist-efficient} for the computational complexity of maintaining this distribution.

At time step $j$, our algorithm will pass a vector $\pj$ into the OCS. After doing so, it updates, for each $w \leq w_{ij}$,
\begin{align*}
    \bx_i(w) \leftarrow \bx_i(w) + (\pj)_i
\end{align*}
For each $w > w_{ij}$, with probability $1 -(1 - (\pj)_i)^r$ we set
\begin{align}
    \label{eq:reset}
    \begin{split}
    \bc_i(w) \leftarrow \bc_i(w) \cdot f(\bx_i(w)) \\
    \bx_i(w) \leftarrow 0 
    \end{split}
\end{align}
Otherwise, we leave $\bc_i(w)$ and $\bx_i(w)$ unchanged. This is the step that makes $\bc_i(w)$ and $\bx_i(w)$ random quantities rather than fixed.

At all steps, we maintain that
\begin{align*}
    y_i(w) = \E[\bc_i(w) \cdot f(\bx_i(w))]
\end{align*}
Note that these updates correspond to \Cref{def:OCS parameters}. For $S = \{j \in R| w_{ij} \geq w\}$, that definition guarantees the probability that $i$ is not matched to a vertex in $S$ is at most $y_i(w)$.

\begin{remark}
    \label{remark:dist-efficient}
    The full support of $(\bx_i(w), \bc_i(w))$ can be exponentially large, and so inefficient to maintain. However, our updates never need full access to the distribution over $(\bx_i(w), \bc_i(w)$. It is sufficient to maintain a distribution over the possible values of $\bx_i(w)$, and for each possible value for $\bx_i(w)=x$, the quantity $\Ex[\bc_i(w) \mid \bx_i(w) = x]$. Since there are only $|R|$ many possible values for $\bx_i(w)$, that information can be stored and operated on efficiently.
\end{remark}

\subsection{Dual updates}
\label{subsec:dual updates}
Our dual updates depend on some function $a: \R_{\geq 0} \to [0, \Gamma]$ with the properties in the below Lemma.
\begin{restatable}[]{lemma}{aproperties}
    \label{lem:a-properties}
    Let $f:\R_{\geq 0} \to [0,1]$ be thrice differentiable, convex, log concave, have a log concave derivative, and satisfy $f(0) = 1$. Define $a:\R_{\geq 0} \to [0,\Gamma] $  as
    \begin{align}
        \label{eq:def-a}
        a(x) \coloneqq f(x) - \int_{0}^\infty e^{-t} f(t + x) dt.
    \end{align}
    For $r$ satisfying \Cref{eq:r-upper-bound-thm} and $\Gamma$ defined in \Cref{eq:def-Gamma}, $a$ satisfies the below properties.
    \begin{align}
        \label{eq:a diffeq} & a'(x) = f'(x) + a(x)  &\forall x \geq 0 \\
        \label{eq:a-decr} &a'(x) \leq 0 &\forall x\geq 0  \\
        &\label{eq:a-gamma}a(0) = \Gamma \\
       \label{eq: f lower a} &f(x) \cdot a(0) = f(x) \cdot \Gamma \leq  a(x)  &\forall x \geq 0 \\
       \label{eq: r upper bound} & r\cdot(a(x) - f(x) \cdot \Gamma) \leq \Gamma - a(x)
       &\forall x \geq 0
    \end{align}
\end{restatable}
We prove \Cref{lem:a-properties} in \Cref{sec:prove-a-properties}. For the remainder of this section, let $a$ be as defined in \Cref{eq:def-a}.

For any $i \in L$ and weight-level $w \geq 0$, we maintain that
\begin{align*}
    \alpha_i(w) = \Ex\Big[\bc_i(w) \cdot a\big(\bx_i(w)\big)\Big],
\end{align*}
where $\bc_i(w),$ $\bx_i(w)$ are the same as in the primal updates. After $\alpha_i$ is updated at time step $j$, we set $\beta_j$ to the minimum value so that approximate dual feasibility holds. Specifically,
\begin{align}
    \label{eq:beta-set}
    \beta_j = \max\Big(0, \max_{i \in L} \, \Gamma \cdot w_{ij} - \alpha_i\Big)
\end{align}
For the above setting of $\beta_j$, dual feasibility will hold for $w_{ij}$ at time step $j$. It will remain satisfied as long $\alpha_i$ never decreases in future time steps. At any time step $j'$, if $(p^{(j')})_i = 0$, then $\alpha_i$ is not changed at that time step. The following proposition guarantees that if $(p^{(j')})_i > 0$, then $\alpha_i$ can only increase, which implies dual feasibility remains satisfied.

\begin{proposition}
    \label{prop:alpha-incr}
    Fix some time step $j$. For each $i \in L$, $w \geq 0$ and $p \in [0,1]$, let $\alpha_i[p]$ and $\alpha_i(w)[p]$ denote the value that would be assigned to $\alpha_i$ and $\alpha_i(w)$ respectively at the end of this time step if $(\pj)_i = p$. Then,
    \begin{enumerate}
        \item $\alpha_i[p]$ is a continuous and increasing function of $p$.
        \item $\alpha_i(w)[p]$ is a continuous and decreasing function of $p$.
    \end{enumerate}
    
\end{proposition}
\begin{proof}
    Recall that
    \begin{align*}
        \alpha_i \coloneqq \int_{0}^\infty (\Gamma - \alpha_i(w))dw,
    \end{align*}
    so it is sufficient to prove that $\alpha_i(w)[p]$ is a continuous and decreasing function of $p$. First, consider the case where $w \leq w_{ij}$. Then, using $\bx_i(w), \bc_i(w)$ as their values \emph{before} the updates from round $j$,
    \begin{align*}
         \frac{d}{dp} \alpha_i(w)[p] = \Ex[\bc_i(w) \cdot a'(\bx_i(w) + p)].
    \end{align*}
    This derivative exists and is at most $0$ by \Cref{eq:a-decr}, so $\alpha_i(w)[0]$ is continuous and decreasing.
    
    Next, consider the case where $w > w_{ij}$. Then,
    \begin{align*}
        \frac{d}{dp} \alpha_i(w)[p] &= \frac{d}{dp}(1 - (1 - p)^r) \cdot \Ex\bigg[\bc_i(w)\cdot  \big( f(\bx_i(w) \cdot a(0) - a(\bx_i(w))\big)\bigg]\\
        &=r \cdot(1-p)^{r-1}\cdot \Ex\bigg[\bc_i(w)\cdot  \big( f(\bx_i(w) \cdot a(0) - a(\bx_i(w))\big)\bigg]\\
        &\leq 0 \tag{\Cref{eq: f lower a}}
    \end{align*}
    Once again, the derivative exists and is less than $0$.
\end{proof}

\subsection{Determining the probability vector}
\label{subsec:determining p}

At time step $j$, our algorithm receives $w_{ij}$ for all $i \in L$. It then sets $\pj$ to a vector satisfying \Cref{lem:choose p}. 

Note that the $\pj$ given by \Cref{lem:choose p} can have a sum of elements that is less than $1$, but \Cref{def:continuous-OCS} requires the sum of the elements to be $1$. To remedy this, we just create a dummy vertex and assign any excess probability to it.
\begin{lemma}
    \label{lem:choose p}
    It is always possible to choose $\pj$ so that, after the updates corresponding to $\pj$,
    \begin{align}
        \label{eq:beta sum}
        \beta_j = \sum_{i \in L} (\pj)_i \cdot (\Gamma \cdot w_{ij} - \alpha_i)
    \end{align}
\end{lemma}

\begin{proof}
    Let $\alpha_i[p]$ be as defined in \Cref{prop:alpha-incr}. For each $i \in L$ and $\beta \geq 0$, let $p_i[\beta]$ be the lowest possible value of $p \geq 0$ such that $\Gamma \cdot w_{ij} - \alpha_i[p] \leq \beta$. We note two properties about this function.
\begin{enumerate}
    \item $p_i[\beta]$ is a continuous and decreasing function of $\beta$. This follows from \Cref{prop:alpha-incr}'s guarantee that $\alpha_i[p]$ is a continuous and increasing function of $p$.
    \item For any $\beta \geq 0$, if $p_i[\beta] \geq 0$, then $\Gamma \cdot w_{ij} - \alpha_i[p] = \beta$. This follows from the continuity of $\alpha_i[p]$.
\end{enumerate}

First, consider the case where $\sum_{i \in L} p_i[0] \leq 1$. In this case, set $(\pj)_i = p_i[0]$ for each $i \in L$. Then, $\beta_j$ will be set to $0$ in \Cref{eq:beta-set} and
\begin{align*}
    \sum_{i \in L} (\pj)_i &\cdot (\Gamma \cdot w_{ij} - \alpha_i) \\
    &= \sum_{i \in L} (\pj)_i \cdot (0) \tag{$p_i[\beta] \geq 0 \implies w_{ij} - \alpha_i[p] = \beta$ applied with $\beta = 0$} \\
    &= 0 = \beta_j.
\end{align*}

Otherwise, $\sum_{i \in L} p_i[0] > 1$. Note that $\lim_{\beta \to \infty} \sum_{i \in L} p_i[\beta] = 0$. By the intermediate value theorem applied to $\beta \mapsto \sum_{i \in L} p_i[\beta]$, there is some choice $\beta^\star$ for which $\sum_{i \in L} p_i[\beta^\star] = 1$.

In this case, set $(\pj)_i = p_i[\beta^\star]$ for each $i \in L$. Then, $\beta_j$ will be set to $\beta^\star$ in \Cref{eq:beta-set} and
\begin{align*}
    \sum_{i \in L} (\pj)_i &\cdot (\Gamma \cdot w_{ij} - \alpha_i) \\
    &= \sum_{i \in L} (\pj_i) \cdot \beta^\star \tag{$p_i[\beta] \geq 0 \implies w_{ij} - \alpha_i[p] = \beta$ applied with $\beta = \beta^\star$} \\
    &= 1 \cdot \beta^\star \tag{$\sum_{i \in L}p_i[\beta^\star] = 1$} \\
    &= \beta^\star = \beta_j.
\end{align*}

\end{proof}

\subsection{Reverse weak duality}
\label{subsec:reverse weak duality}
Since we already have that approximate dual feasibility holds, we only need to prove that reverse weak duality holds to complete the proof that this algorithm is $\Gamma$-competitive. At the start of our algorithm we have that $P = D = 0$. Let $\Delta_j P$ and $\Delta_j D$ be the amount the primal and dual respectively increase after the update corresponding to $j \in R$. We will prove that $\Delta_j P \geq \Delta_j D$ which guarantees reverse weak duality holds.

\begin{lemma}
    For any $j \in R$, if $\pj$ is set as in \Cref{lem:choose p}, then $\Delta_j P \geq \Delta_j D$.
\end{lemma}

\begin{proof}
    For each primal and dual variable, we use a superscript of $^{(j)}$ to refer to its value after the $j^{\mathrm{th}}$ update. Similarly, a superscript of $^{(j - 1)}$ refers to its value before the update. For any variable $\mathrm{var}$, we use the shorthand $\Delta_j \mathrm{var}$ to refer $\mathrm{var}^{(j)} - \mathrm{var}^{(j-1)} $.

    \begin{align*}
        \Delta_j P - \Delta_j D &= \left(-\sum_{i \in L} \int_0^\infty \Delta_j y_i(w) dw\right) - \left(-\sum_{i \in L}\int_{0}^\infty \Delta_j\alpha_i(w) dw + \beta_j\right) \\
        &=  \sum_{i \in L} \left(-\int_0^\infty (\Delta_j y_i(w) -  \Delta_j\alpha_i(w)) dw\right) - \left(\sum_{i \in L} (\pj)_i \cdot (\Gamma \cdot w_{ij} - \alpha_i^{(j)}) \right)
    \end{align*}
    Recall that $\alpha_i = \int_0^\infty (\Gamma - \alpha_i(w))dw$. Therefore,
    \begin{align*}
        \Gamma \cdot w_{ij} - \alpha_i^{(j)} &= \left(\int_0^{w_{ij}} \Gamma \cdot dw\right) - \left(\int_0^\infty (\Gamma - \alpha_i^{(j)}(w))dw \right) \\
        &= \left(\int_0^{w_{ij}} \alpha_i^{(j)}(w) dw\right) - \left(\int_{w_{ij}}^\infty (\Gamma - \alpha_i^{(j)}(w))dw \right)
    \end{align*}
    Combining the above equations, we have that
    \begin{align*}
        \Delta_j P - \Delta_j D = &\sum_{i \in L} \left( \int_{0}^{w_{ij}} \left(- \Delta_j y_i(w) + \Delta_j\alpha_i(w) - (\pj)_i  \cdot \alpha_i^{(j)}(w)\right)dw\right) \\ +&\sum_{i \in L}\left(\int_{w_{ij}}^{\infty} \left(- \Delta_j y_i(w) + \Delta_j\alpha_i(w) + (\pj)_i \cdot \big(\Gamma - \alpha_i^{(j)}(w)\big)\right)dw \right)
    \end{align*}

   We will show that for any choice of $w \geq 0$ and $i \in L$, the corresponding term in the above equation is positive. This implies that $\Delta_j P - \Delta_j D$ is also positive.

    \pparagraph{Case 1: $w \leq w_{ij}$}. We wish to show that $- \Delta_j y_i(w) + \Delta_j\alpha_i(w) - (\pj)_i  \cdot\alpha_i^{(j)}(w)\geq 0$. We expand each term separately.
    \begin{align*}
        \Delta_j y_i(w) &= y_i^{(j)}(w) - y_i^{(j-1)}(w) \\
        &= \Ex\left[\bc_i^{(j)}(w) \cdot (f(\bx_i^{(j)}(w) - f(\bx_i^{(j)}(w) - (\pj)_i))\right] \\
        &= \Ex\left[\int_{\bx_i^{(j)}(w) - (\pj)_i}^{\bx_i^{(j)}(w)} \bc_i^{(j)}(w) \cdot f'(x) dx\right]
    \end{align*}
    Similarly,
    \begin{align*}
        \Delta_j \alpha_i(w) &= \alpha_i^{(j)}(w) - \alpha_i^{(j-1)}(w) \\
        &= \Ex\left[\bc_i^{(j)}(w) \cdot (a(\bx_i^{(j)}(w) - a(\bx_i^{(j)}(w) - (\pj)_i))\right] \\
        &= \Ex\left[\int_{\bx_i^{(j)}(w) - (\pj)_i}^{\bx_i^{(j)}(w)} \bc_i^{(j)}(w) \cdot a'(x) dx\right]
    \end{align*}
    Finally,
    \begin{align*}
        (\pj)_i  \cdot\alpha_i^{(j)}(w) &= (\pj)_i \cdot \Ex\left[\bc_i^{(j)}(w) \cdot a(\bx_i^{(j)}(w)) \right] \\
        &= \int_{\bx_i^{(j)}(w) - (\pj)_i}^{\bx_i^{(j)}(w)} \Ex\left[\bc_i^{(j)}(w) \cdot a(\bx_i^{(j)}(w)) \right]dx \\
        &=  \Ex\left[\int_{\bx_i^{(j)}(w) - (\pj)_i}^{\bx_i^{(j)}(w)} \bc_i^{(j)}(w) \cdot a(\bx_i^{(j)}(w)) dx\right]
    \end{align*}
    Comparing the above three equations, it is enough to show that for any $x < x_i^{(j)}(w)$, that $-f'(x) + a'(x) - a(x_i^{(j)}) \geq 0$.
    \begin{align*}
        -f'(x) + a'(x) - a(x_i^{(j)}) &\geq -f'(x) + a'(x) - a(x) \tag{$a$ is decreasing} \\
         &= 0 \tag{\Cref{eq:a diffeq}}
    \end{align*}
    This completes the first case.
    
    \pparagraph{Case 2: $w > w_{ij}$}. We wish to show that $- \Delta_j y_i(w) + \Delta_j\alpha_i(w) + (\pj)_i \cdot \big(\Gamma - \alpha_i^{(j)}(w)\big) \geq 0$. For $w > w_{ij}$, $y_i(w)$ does not change. Hence, we only need to expand the other two terms. Recall that, for $w \geq w_{ij}$, we set $\bc_i(w) \leftarrow \bc_i(w) \cdot f(\bx_i(w))$ and $\bx_i(w) \leftarrow 0$ with probability $1 - (1 - (\pj)_i)^r$. Otherwise, we leave the parameters unchanged. Expanding,
    \begin{align*}
        \Delta_j\alpha_i(w) = \Ex\left[(1 - (1 - (\pj)_i)^r) \cdot \bc_i^{(j-1)}(w) \cdot \big(f(\bx_i^{(j-1)}(w)) \cdot a(0) - a(\bx_i^{(j-1)}(w))\big)\right].
    \end{align*}
    Recall from \Cref{eq: f lower a} that $f(x) \cdot a(0) \leq a(x)$ for all $x \geq 0$. Therefore, we may replace the $1 - (1 - (\pj)_i)^r$ factor with an upper bound on it, in this case $r \cdot (\pj)_i$, and bound,
    \begin{align*}
        \Delta_j\alpha_i(w) \geq \Ex\left[r \cdot (\pj)_i \cdot \bc_i^{(j-1)}(w) \cdot \big(f(\bx_i^{(j-1)}(w)) \cdot a(0) - a(\bx_i^{(j-1)}(w))\big)\right].
    \end{align*}
    Using the fact that $\alpha_i^{(j)}(w) \leq \alpha_i^{(j-1)}(w)$, which follows from \Cref{prop:alpha-incr}, we expand the other term of our desired inequality.
    \begin{align*}
        (\pj)_i \cdot \big(\Gamma - \alpha_i^{(j)}(w)\big) &\geq (\pj)_i \cdot \big(\Gamma - \alpha_i^{(j-1)}(w)\big) \\ &= (\pj)_i \cdot \Ex\left[\Gamma -   \bc_i^{(j-1)}(w) \cdot a(\bx_i^{(j-1)}(w))\right] 
    \end{align*}
    The desired inequality holds as long as, for all $x \geq 0$,
    \begin{align*}
        r \cdot (f(x) \cdot a(0) - a(x)) + (\Gamma - a(x)) \geq 0.
    \end{align*}
    which is equivalent to \Cref{eq: r upper bound}. This completes the proof of the second case and reverse weak duality.
\end{proof}

\Crefname{enumi}{Step}{Steps}
\begin{figure}
  \captionsetup{width=.9\linewidth}
\begin{tcolorbox}[colback = white,arc=1mm, boxrule=0.25mm]
\vspace{3pt} 
$\textsc{MatchOnline}(L)$
\begin{enumerate}
    \item \textbf{Initialization:} For each $i \in L$, and $w \geq 0$, initialize the distributions $\bx_i(w), \bc_i(w)$ as
    \begin{align*}
        \bx_i(w) \overset{\mathrm{init}}{\longleftarrow} 0 \quad \text{with probability $1$}\\
        \bc_i(w) \overset{\mathrm{init}}{\longleftarrow} 1 \quad \text{with probability $1$}
    \end{align*}
    The algorithm maintains the \emph{distribution} of possible $(\bx_i(w), \bc_i(w))$ over the randomness in \Cref{step:update-dist}.
    \item \textbf{Invariants:} \label{step:invariants}For each $i \in L$ and $w \geq 0$, maintain the primal variables, where the expectations are over the distribution of $(\bx_i(w), \bc_i(w))$,
    \begin{align*}
        y_i(w) = \E[\bc_i(w) \cdot f(\bx_i(w))]
    \end{align*}
    and dual variables
    \begin{align*}
        \alpha_i(w) = \Ex\Big[\bc_i(w) \cdot a\big(\bx_i(w)\big)\Big],
    \end{align*}
    where $a$ is the function from \Cref{lem:a-properties}.
    \item Upon receiving $w_{ij}$ for each $i \in L$
    \begin{enumerate}
        \item \textbf{Determine input to OCS:} Using \Cref{lem:choose p}, determine $\pj$ to satisfy \Cref{eq:beta sum}. This choice of $\pj$ is a function of the current state of the algorithm as well as $w_{ij}$ for each $i \in L$.
        \item \textbf{Update distribution of $(\bx_i(w), \bc_i(w))$:} \label{step:update-dist}For each $i \in L$ and $w \leq w_{ij}$, update
        \begin{align*}
            \bx_i(w) \leftarrow \bx_i(w) + (\pj)_i.
        \end{align*}
        For each $w > w_{ij}$, with probability $1 -(1 - (\pj)_i)^r$ set
        \begin{align*}
            \begin{split}
            \bc_i(w) \leftarrow \bc_i(w) \cdot f(\bx_i(w)) \\
            \bx_i(w) \leftarrow 0 
            \end{split}
        \end{align*}
        Otherwise, leave $\bc_i(w)$ and $\bx_i(w)$ unchanged. This is the step that makes $\bc_i(w)$ and $\bx_i(w)$ random variables rather than fixed.
        \item \textbf{Invariants:} Maintain the primal and dual variables as in \Cref{step:invariants}.
        \item \textbf{Match:} Pass $\pj$ into the OCS and match $j$ to the winning vertex.
    \end{enumerate}
\end{enumerate}
\end{tcolorbox}
\caption{The algorithm for edge-weighted online bipartite matching. For computational efficiency concerns, see \Cref{remark:step function,remark:dist-efficient}.}
\label{fig:MatchOnline}
\end{figure}
\Crefname{enumi}{Item}{Item}

\section{From discrete to continuous OCSes}
\label{sec:cont-to-disc-OCS}

In this section, we show how to construct a continuous OCS from a discrete OCS. That will allow us to prove \Cref{cor:Gamma-fahrbach,thm:discrete-competitive-ratio-informal}.

\begin{lemma}
    \label{lem:disc-to-cont}
    If there is an $(F, m)$-discrete OCS for convex $F$, there is an $(f, r=m)$-continuous OCS for
    \begin{align}
        \label{eq:disc-to-cont}
        f(x) \coloneqq \Ex_{\bk \sim \Poi(\lambda = m \cdot x)}[F(\bk)]
    \end{align}
\end{lemma}
The continuous OCS algorithm in \Cref{lem:disc-to-cont} is simple. When the continuous OCS receives the probability vector $\pj$, it sets the multiset $A_j = \{\bi_1, \ldots, \bi_m\}$ where each $\bi_\ell$ is sampled iid from $\Categorical(\pj)$. It passes this multiset into the $(F,m)$-discrete OCS, and then selects the same vertex selected by the discrete OCS. We include pseudocode for this reduction in \Cref{fig:disc-from-cont}.

\begin{figure}[h]
  \captionsetup{width=.9\linewidth}
\begin{tcolorbox}[colback = white,arc=1mm, boxrule=0.25mm]
\vspace{3pt} 
$\mathrm{ContinuousOCS}(\mathcal{O}, m)$
\begin{itemize}
    \item[] Upon receiving $\pj$
    \begin{enumerate}
        \item \textbf{Sample:} For each $k = 1, \ldots, m$, sample independently $\bi_k \sim \mathrm{Categorical(\pj)}$.
        \item \textbf{Use discrete OCS:} Pass the mulitiset $A_j = \{\bi_1, \ldots, \bi_m\}$ into $\mathcal{O}$ and return the same winner as it.
    \end{enumerate}
\end{itemize}
\end{tcolorbox}
\caption{A construction for a continuous OCS using black-box access to a $m$-discrete OCS $\mathcal{O}$.}
\label{fig:disc-from-cont}
\end{figure}

Before analyzing this OCS, we'll need some well-known facts about the convex order (recall \Cref{def:cx}).
\begin{fact}[Theorem 3.A.4 of \cite{SS07}]
    \label{fact:cx-coupling}
    For any univariate random variables $\ba, \bb$, if there exists a coupling $(\hat{\ba}, \hat{\bb})$ where $\ba \overset{d}{=} \hat{\ba} $ and $\bb \overset{d}{=} \hat{\bb}$ such that
    \begin{align*}
        \Ex[\hat{\bb} \mid \hat{\ba}] = \hat{\ba}
    \end{align*}
    then $\ba \cx \bb$.
\end{fact}

\begin{fact}[Theorem 3.A.12(d) of \cite{SS07}]
    \label{fact:cx-closed-sum}
    The convex order is closed under convolution: Let $\ba_1, \ldots, \ba_k$ and $\bb_1, \ldots, \bb_k$ each be independent univariate random variables. If $\ba_\ell \cx \bb_\ell$ for each $\ell = 1, \ldots, m$, then
    \begin{align*}
        \sum_{\ell = 1}^m \ba_\ell \cx \sum_{\ell=1}^m \bb_\ell.
    \end{align*}
\end{fact}

The following proposition is a simple consequence of the above two facts.
\begin{proposition}
    \label{prop:ber-sum}
    For any $p_1, \ldots, p_n \in [0,1]$, let $\bx_i \sim \Ber(p_i)$ and $\by \sim \Poi(\lambda = \sum_{i = 1}^n p_i)$. Then,
    \begin{align*}
        \sum_{i=1}^n \bx_i \cx \by.
    \end{align*}
\end{proposition}
\begin{proof}
    By \Cref{fact:cx-closed-sum}, it is enough to prove that $\bx_i \cx \by_i$ where $\by_i \sim \Poi(\lambda = p_i)$. First, we note that $\by_i$ is more likely to be $0$ than $\bx_i$.
    \begin{align*}
        \Pr[\by_i = 0] &= \frac{p_i^k \cdot e^{-{p_i}}}{0!} = e^{-(p_i)} \geq (1 - p_i) = \Pr[\bx_i = 0].
    \end{align*}
    Consider the following coupling: First draw $\hat{\by}_i \sim \Poi(p_i)$. If $\hat{\by_i} = 0$, then with probability $(\Pr[\bx_i = 0])/(\Pr[\by_i = 0])$, set $\hat{\bx_i} = 0$ as well. Otherwise, set $\hat{\bx_i} = 1$. 
    
    We verify that this coupling meets the criteria of \Cref{fact:cx-coupling} which implies $\bx_i \cx \by_i$. By design, we have that $\hat{\by_i} \overset{d}{=} \by_i$ and $\hat{\bx_i} \overset{d}{=} \bx_i$. If $\hat{\bx_i} = 0$, then $\hat{\by_i} = 0$ as well, so
    \begin{align*}
        \Ex[\hat{\by_i} \mid \hat{\bx_i} = 0] = 0
    \end{align*}
    Then, using the law of total expectation, we have that
    \begin{align*}
        \Ex[\hat{\by_i} \mid \hat{\bx_i} = 0] \cdot \Pr[\hat{\bx_i} = 0] + \Ex[\hat{\by_i} \mid \hat{\bx_i} = 1] \cdot \Pr[\hat{\bx_i} = 1] = \Ex[\hat{\by_i}]
    \end{align*}
    Substituting $\Ex[\hat{\by_i} \mid \hat{\bx_i} = 0] = 0$, $\Ex[\hat{\by_i}] = p_i$, and $\Pr[\hat{\bx_i} = 1] = p_i$, we have that
    \begin{align*}
         \Ex[\hat{\by_i} \mid \hat{\bx_i} = 1] = 1.
    \end{align*}
    Therefore, by \Cref{fact:cx-coupling}, $\bx_i \cx \by_i$.
\end{proof}

Now, we are ready to prove \Cref{lem:disc-to-cont}.
\begin{proof}
    Let $\mathcal{A}$ be some $(F, m)$ discrete OCS. We construct a continuous OCS $\mathcal{B}$ as in \Cref{fig:disc-from-cont}: Whenever $\mathcal{B}$ receives the probability vector $\pj$, it sets $\bA_j = \{\bi_1, \ldots, \bi_m\}$ where each $\bi_\ell$ is sampled iid from $\Categorical(\pj)$. It passes this multiset into $\mathcal{A}$, and selects the same winner that $\mathcal{A}$ selected. By \Cref{def:discrete-parameters}, for any disjoint consecutive sequences of time steps $S_1, \ldots, S_k$,
    \begin{align*}
        \Pr[i \text{ never selected in $S_1 \cup \ldots \cup S_k$}] \leq \prod_{\ell = 1}^k F(\textbf{count}_i(S_\ell))
    \end{align*}
    where $\textbf{count}_i(S_\ell)$ is a random variable defined as
    \begin{align*}
        \textbf{count}_i(S_\ell) \coloneqq \sum_{j \in S_{\ell}}(\br_j(i)) \quad \quad \text{where $\br_j(i)$ is the multiplicity of $i$ in $\bA_j$}.
    \end{align*}
    Based on how we defined $\bA_j$, we have that $\br_j(i)$ is the sum of $m$ independent variables distributed according to $\Ber((\pj)_i)$. Therefore, by \Cref{prop:ber-sum},
    \begin{align*}
        \textbf{count}_i(S_\ell) \cx \bk_{\ell} \quad \quad \text{where } \bk_{\ell} \sim \Poi\left(m \cdot \sum_{j \in S_{\ell}} (\pj)_i \right).
    \end{align*}
    
    We are ready to show that $\mathcal{B}$ meets the requirements of \Cref{def:OCS parameters}. First, we consider some consecutive $S = \{j_1, \ldots, j_2\}$. Since $F$ is convex,
    \begin{align*}
        \Pr[i \text{ not selected in }S] &= \Ex[F(\textbf{count}_i(S))] \\
        &\leq \Ex[F(\bk)] \quad \quad \text{where } \bk \sim \Poi\left(m \cdot \sum_{j \in S} (\pj)_i \right) \tag{$\textbf{count}_i(S) \cx \bk$} \\
        &= f\left(  \sum_{j \in S} (\pj)_i\right) \tag{\Cref{eq:disc-to-cont}}
    \end{align*}
    For the second requirement, ``partially ignore small gaps," of \Cref{def:OCS parameters}, consider any disjoint $S_1, S_2 \subseteq \N$ where every element of $S_1$ precedes every element of $S_2$. Define $j^\star$ as in \Cref{def:OCS parameters}. Then, the probability $\br_{j^\star}(i) = 0$ is $w = (1 - (p^{(j^\star)})_i)^m$. If $\br_{j^\star}(i) = 0$ then the probability $i$ is selected in $S_1 \cup S_2$ is the same as if the time step $j^\star$ didn't exist (since $\mathcal{A}$ won't see $i$ in it). Hence, \Cref{eq:ignore-gaps} holds.
\end{proof}

In order to give a formal version of \Cref{thm:discrete-competitive-ratio-informal} we'll need to know the derivatives of $f$. There is a convenient relation between the \emph{discrete derivatives} of $F$ and the \emph{derivatives} of $f$.

\begin{lemma}
    \label{lem:cont-to-disc-deriv}
    Choose any $F: \Z_{\geq 0} \to \R$, $m \in Z$ and let $f:\R_{\geq 0} \to \R$ be
    \begin{align*}
        f(x) = \Ex_{\bk \sim \Poi(mx)}[F(\bk)].
    \end{align*}
    Then, for any $\ell \in \N$ and $x \in \R_{\geq 0}$, the derivatives satisfy the relation
    \begin{align*}
        f^{(\ell)}(x) = m^\ell \cdot \Ex_{\bk \sim \Poi(mx)}[\Delta^{(\ell)} F(\bk)].
    \end{align*}
\end{lemma}
\begin{proof}
    It is enough to prove the case where $\ell = 1$, as larger $\ell$ follow by induction. The remainder of this proof is algebraic manipulations. We expand
    \begin{align*}
        f'(x) &= \frac{d}{dx}\left(\sum_{k = 0}^\infty \frac{e^{-mx} \cdot (mx)^k}{k!} \cdot F(k) \right) \\
        &=  \sum_{k = 0}^\infty \left(\frac{d}{dx} e^{-mx} x^k \right)\cdot \frac{m^k \cdot F(k)}{k!} 
    \end{align*}
    We'll use the linearity of derivatives. To do so, we first compute the derivative of the above coefficients.
    \begin{align*}
        \frac{d}{dx} e^{-mx} x^k = \begin{cases}
        -me^{-mx} \cdot x^k & k = 0 \\
        -me^{-mx} \cdot x^k + ke^{-mx}x^{k-1} & k \geq 1
        \end{cases}
    \end{align*}
    Therefore,
    \begin{align*}
        f'(x) &= -m \cdot \sum_{k = 0}^\infty \frac{e^{-mx} \cdot (mx)^k}{k!} \cdot F(k) + \sum_{k = 1}^\infty \frac{ke^{-mx}(mx)^{k-1}\cdot m}{k!} \cdot F(k)\\
         &= m \cdot\left(-\sum_{k = 0}^\infty \frac{e^{-mx} \cdot (mx)^k}{k!} \cdot F(k) + \sum_{k = 0}^\infty \frac{e^{-mx}(mx)^{k}\cdot m}{k!} \cdot F(k+1)\right) \\
        &= m \cdot\left(\sum_{k = 0}^\infty \frac{e^{-mx} \cdot (mx)^k}{k!} \cdot (F(k+1) - F(k))\right) \\
        &= \Ex_{\bk \sim \Poi(mx)}[\Delta F(\bk)],
    \end{align*}
    as desired.
\end{proof}

As a consequence, we give the main theorem of this subsection, a formal version of \Cref{thm:discrete-competitive-ratio-informal}
\begin{theorem}[Formal version of \Cref{thm:discrete-competitive-ratio-informal}]
    \label{thm:discrete-competitive-ratio-formal}
    For any $m \in \N$ and convex $F: \Z_{\geq 0} \to [0,1]$ satisfying, for any $k_1, k_2 \in \Z_{\geq 0}$,
    \begin{align}
        \label{eq:discrete-log-concave-regular}
        \Delta^{(2)} F(k_1) \cdot F(k_2) + F(k_1) \cdot \Delta^{(2)}  F(k_2) \leq 2 \cdot \Delta F(k_1) \cdot \Delta F(k_2)\\
        \label{eq:discrete-log-concave}
        \Delta^{(3)} F(k_1) \cdot \Delta F(k_2) + \Delta F(k_1) \cdot \Delta^{(3)}  F(k_2) \leq 2 \cdot \Delta^{(2)}F(k_1) \cdot \Delta^{(2)}F(k_2)
    \end{align}
    Let
    \begin{align}
        \label{eq:Gamma-disc}
        \Gamma \coloneqq 1 - \sum_{n = 0}^\infty \frac{m^n}{(m+1)^{n+1}} \cdot F(n),
    \end{align}
    and suppose that
    \begin{align}
        \label{eq:m-bound}
        m \leq \frac{-m \cdot \Delta F(0) - \Gamma}{\Gamma + (1 - \Gamma) \cdot m \cdot \Delta F(0)}.
    \end{align}
    Then, if there an $(F,m)$-discrete OCS, there is an $\Gamma$-competitive algorithm for edge-weighted online bipartite matching.
\end{theorem}
\begin{proof}
    Apply \Cref{lem:disc-to-cont} to construct an $(f,m)$-continuous OCS for $f$ defined \Cref{eq:disc-to-cont}. We wish to apply \Cref{thm:OCS-cont-Gamma} to that continuous OCS. If $F$ is convex, then $f$ is convex by \Cref{lem:cont-to-disc-deriv}. We verify that $f$ and $f'$ are log-concave:
    \begin{align*}
        f''(x) \cdot f(x) - (f'(x))^2 = \Ex_{\bk_1, \bk_2 \sim \Poi(mx)}\left[\Delta^{(2)} F(k_1) \cdot  F(k_2) -   \Delta F(k_1) \cdot \Delta F(k_2)\right] \leq 0.\\
        f'''(x) \cdot f'(x) - (f''(x))^2 = \Ex_{\bk_1, \bk_2 \sim \Poi(mx)}\left[\Delta^{(3)} F(k_1) \cdot \Delta F(k_2) -   \Delta^{(2)}F(k_1) \cdot \Delta^{(2)}F(k_2)\right] \leq 0.
    \end{align*}
    Therefore, the desired result holds with
    \begin{align*}
        \Gamma = 1 - \int_{0}^{\infty}e^{-t} \Ex_{\bk \sim \Poi(mt)}[F(\bk)] dt
    \end{align*}
    We just need to show the above simplifies to \Cref{eq:Gamma-disc}.
    \begin{align*}
        1 - \Gamma &= \int_{0}^{\infty}e^{-t} \sum_{k = 0}^\infty \Pr[\Poi(m\cdot t) = k] \cdot F(k)dt\\
        &= \int_{0}^{\infty}e^{-t} \sum_{k = 0}^\infty \frac{e^{-mt} \cdot (mt)^k}{k!} \cdot F(k)dt \\
        &= \sum_{k = 0}^\infty \frac{m^k}{k!} F(k) \cdot \int_{0}^{\infty}e^{-(m + 1)t} \cdot t^k dt \\
        &= \sum_{k = 0}^\infty \frac{m^k}{k!} F(k) \cdot \left(\frac{k!}{(m+1)^{k+1}}\right) \\
        &= \sum_{k = 0}^\infty \frac{m^k}{(m+1)^{k+1}} \cdot F(k).
    \end{align*}
    Hence, \Cref{thm:OCS-cont-Gamma} guarantees the same competitive ratio as in \Cref{thm:discrete-competitive-ratio-formal}.
\end{proof}
\begin{remark}
    \label{rem:m-constraint}
    \Cref{eq:m-bound} constrains how large $m$ can be (we should think of $m \cdot \Delta F(0)$ as fixed; it is $-1$ in our constructions). For that reason, we are forced to set $m \leq 6$ in order to achieve a competitive ratio of more than $0.533$. Depending on the construction, the constraint in \Cref{eq:m-bound} may not be fully necessary. That constraint comes from the fact that in \Cref{lem:disc-to-cont}, the continuous OCS is $(f,r)$-continuous for $r = m$. There may be tighter bounds for $r$ that allow $m$ to be set significantly higher. For our constructions, we believe even in the case of $m \to \infty$, that $r$ approaches some constant. Since \Cref{eq:Gamma-disc} is often increasing in $m$ (it is for various constructions we tried), more tight analysis of $r$ coupled with larger $m$ could lead to an improved competitive ratio.
\end{remark}

\subsection{Proof of \Cref{cor:Gamma-fahrbach,thm:main}}

We first prove the following more general result, and then use it for both \Cref{cor:Gamma-fahrbach,thm:main}.
\begin{lemma}
    \label{lem:disc-OCS-competitive-exponential}
    For any $m, k^\star \in \N$, $c \in [0,1]$, and $F:\Z_{\geq 0} \to [0,1]$ where, for all $k \geq k^\star$,
    \begin{align}
        \label{eq:F-k-star}
        F(k) = F(k^{\star}) \cdot (1 - c)^{k - k^{\star}}.
    \end{align}
    Suppose that \Cref{eq:discrete-log-concave-regular,eq:discrete-log-concave} hold for all $k_1, k_2 < k^\star$, that for all $k < k^\star$, that $\Delta^2 F(k) \geq 0$ and
    \begin{align}
        \label{eq:discrete-log-concave-c-regular}
        \Delta^{(2)} F(k) + 2c \cdot \Delta F(k) + c^2\cdot F(k) \leq 0\\
        \label{eq:discrete-log-concave-c}
        \Delta^{(3)} F(k) + 2c \cdot \Delta^{(2)} F(k) + c^2\cdot \Delta F(k) \geq 0.
    \end{align}
    Let,
    \begin{align}
        \label{eq:competitive-ratio-exponential}
        \Gamma \coloneqq 1 - \left(\sum_{n = 0}^{k^\star - 1} \frac{m^n}{(m+1)^{n+1}} \cdot F(n)\right) - \left(\frac{m}{m+1}\right)^{k^\star} \cdot \frac{1}{1 + m\cdot c} \cdot F(k^{\star}),
    \end{align}
    and suppose that
    \begin{align}
        \label{eq:m-bound-exponential}
        m \leq \frac{-m \cdot \Delta F(0) - \Gamma}{\Gamma + (1 - \Gamma) \cdot m \cdot \Delta F(0)}.
    \end{align}
    Then, if there is an $(F,m)$-discrete OCS, there is an $\Gamma$-competitive algorithm for edge-weighted online bipartite matching.
\end{lemma}
\begin{proof}
     To verify the conditions of \Cref{thm:discrete-competitive-ratio-formal}, we need to compute the derivatives of $F$ in the case where $k \geq k^\star$.
    \begin{align*}
        \Delta^{(\ell)} F(k) = (-c)^\ell \cdot F(k) \quad \quad \text{for }k \geq k^\star
    \end{align*}
    Therefore, the condition $\Delta^{(2)}F(k) \geq 0$ for $k < k^\star$ is enough to guarantee that $F$ is convex. 
    
    We need to show \Cref{eq:discrete-log-concave,eq:discrete-log-concave-regular} holds in the case where exactly one of $k_1$ and $k_2$ is at least $k^\star$ and in the case where both are at least $k^\star$. In the case where $k_1, k_2 \geq k^{\star}$, \Cref{eq:discrete-log-concave-regular} simplifies to
    \begin{align*}
        c^2 \cdot F(k_1) \cdot F(k_2) +  c^2 \cdot F(k_1)  \cdot  F(k_2) \leq 2 c^2 \cdot F(k_1) \cdot F(k_2)
    \end{align*}
    and \Cref{eq:discrete-log-concave} to
    \begin{align*}
        c^4 \cdot F(k_1) \cdot F(k_2) +  c^4 \cdot F(k_1)  \cdot  F(k_2) \leq 2 c^4 \cdot F(k_1) \cdot F(k_2)
    \end{align*}
    both of which hold with equality. Consider the other case, where exactly one of $k_1, k_2$ is less than $k^\star$. Without loss of generality, let it be $k_2$.
    Then, \Cref{eq:discrete-log-concave-regular} requires
    \begin{align*}
        c^2F(k_1)F(k_2) + F(k_1)\Delta^{(2)}F(k_2) \leq 2 (-c)F(k_1) \Delta F(k_2)
    \end{align*}
    which is equivalent to \Cref{eq:discrete-log-concave-c-regular}. Similarly \Cref{eq:discrete-log-concave} requires
    \begin{align*}
        -c^3 \cdot  F(k_1) \cdot \Delta F(k_2) - c \cdot F(k_1) \cdot  \Delta^{(3)} F(k_2) \leq 2 c^2 \cdot F(k_1) \cdot \Delta^{(2)}F(k_2)
    \end{align*}
    which is equivalent to \Cref{eq:discrete-log-concave-c}. Therefore, we can apply \Cref{thm:discrete-competitive-ratio-formal} to determine the competitive ratio. First, we compute
    \begin{align*}
        \sum_{n = k^\star}^{\infty}\frac{m^n}{(m+1)^{n+1}} \cdot F(n) &= \sum_{n = k^\star}^{\infty}\frac{m^n}{(m+1)^{n+1}} \cdot F(k^\star) \cdot(1 - c)^{n - k^\star} \\
        &= \frac{\frac{m^{k^\star}}{(m+1)^{k^{\star}+1}} \cdot F(k^\star)}{1 - \frac{m}{m+1} \cdot (1 - c)} \\
        &=\left(\frac{m}{m+1}\right)^{k^\star}  \cdot \frac{1}{(m+1) - m \cdot(1 - c)}\cdot F(k^\star) \\
         &=\left(\frac{m}{m+1}\right)^{k^\star}  \cdot \frac{1}{1 + m\cdot c}\cdot F(k^\star)
    \end{align*}
    By \Cref{eq:Gamma-disc}, the competitive ratio is,
    \begin{align*}
        \Gamma &= 1 - \sum_{n = 0}^\infty \frac{m^n}{(m+1)^{n+1}} \cdot F(n) \\
        &= 1 - \left(\sum_{n = 0}^{k^\star - 1} \frac{m^n}{(m+1)^{n+1}} \right) - \left(\frac{m}{m+1}\right)^{k^\star}  \cdot \frac{1}{1 + m\cdot c}\cdot F(k^\star).
    \end{align*}
\end{proof}

Next we prove the following theorem, restated for convenience.

\corfahr* 
\begin{proof}
    Recall that a $\gamma$-OCS is $(F, 2)$-discrete OCS for $F$,
    \begin{align*}
        F(k) \coloneqq 2^{-k} \cdot (1 - \gamma)^{(k-1)_+}.
    \end{align*}
    $F$ has the form in \Cref{eq:F-k-star} for $k^\star = 1$ and $c = \frac{1 + \gamma}{2}$. We compute discrete derivatives for $k = 0$:
    \begin{align*}
        F(0) &= 1\\
        \Delta F(0) &= 
        -\lfrac{1}{2} \\
        \Delta^{(2)} F(0) &= \frac{1 - \gamma}{4} \\
        \Delta^{(3)} F(0) &=\frac{\gamma^2 + 4 \gamma - 1}{8}
    \end{align*}
    The second derivative is positive for all $\gamma \in (0,1)$, so $F$ is convex. We verify \Cref{eq:discrete-log-concave-regular,eq:discrete-log-concave} for all $k_1, k_2 < k^\star$. The only such case is when $k_1 = k_2 = 0$ in which case we need the following quantities to be nonpositive.
    \begin{align*}
        2 \cdot \frac{1 - \gamma}{4} - 2 \left(-\frac{1}{2}\right)^2 = -\frac{\gamma}{2} \leq 0. \\
        2 \cdot \left(-\frac{1}{2}\right) \cdot \frac{\gamma^2 + 4 \gamma - 1}{8} - 2 \cdot \left(\frac{1 - \gamma}{4}\right)^2 = -\frac{\gamma(\gamma + 1)}{4} \leq 0.
    \end{align*}
    We verify \Cref{eq:discrete-log-concave-c-regular,eq:discrete-log-concave-c} for $k = 0$
    \begin{align*}
        \frac{1 - \gamma}{4} + 2\left( \frac{1 + \gamma}{2}\right) \left(-\frac{1}{2}\right) + \left( \frac{1 + \gamma}{2}\right)^2 \cdot 1 = \frac{\gamma \cdot(\gamma - 1)}{4} \leq 0 \\
       \frac{\gamma^2 + 4 \gamma - 1}{8} + 2\left( \frac{1 + \gamma}{2}\right) \cdot \frac{1 - \gamma}{4} + \left( \frac{1 + \gamma}{2}\right)^2\cdot (-\frac{1}{2}) = \frac{\gamma(1 - \gamma)}{4} \geq 0 
    \end{align*}
    Then, we compute the competitive ratio.
    \begin{align*}
        \Gamma &= 1 - \left(\sum_{n = 0}^{k^\star - 1} \frac{m^n}{(m+1)^{n+1}} \cdot F(n)\right) - \left(\frac{m}{m+1}\right)^{k^\star} \cdot \frac{1}{1 + m\cdot c} \cdot F(k^{\star}) \\
        &= 1 - \frac{1}{3} \cdot 1 - \frac{2}{3}\cdot \frac{1}{2 + \gamma} \cdot \frac{1}{2}\\
        &= \frac{3 + 2\gamma}{6 + 3\gamma}
    \end{align*}
    Lastly, we need to verify \Cref{eq:m-bound-exponential}. Since $\Delta F(0) = -\lfrac{1}{m}$, the upper bound on $m$ simplifies to
    \begin{align*}
        \frac{1 - \Gamma}{\Gamma - (1 - \Gamma)} = \frac{1 - \Gamma}{2\Gamma - 1} 
    \end{align*}
    This is a decreasing function of $\Gamma$. The largest possible value for $\Gamma$ is $\frac{5}{9}$, corresponding to when $\gamma = 1$. In the case, the above simplifies to $4$. Since $m = 2 \leq 4$, \Cref{eq:m-bound-exponential} holds.
\end{proof}

\Cref{thm:main} is a consequence of \Cref{lem:disc-OCS-competitive-exponential} and the existence of an $(F,6)$-discrete OCS for the $F:\Z_{\geq 0} \to [0,1]$ described in \Cref{fig:OCS-parameters}. In a publicly available Colab notebook\footnote{To see our code, go to \url{https://colab.research.google.com/drive/1yQErphKVkwwPPXsUWGT2b-nIPPaLBtnh?usp=sharing}}, we verify that OCS has the properties required by \Cref{lem:disc-OCS-competitive-exponential} and that the competitive ratio computed by \Cref{eq:competitive-ratio-exponential} is $\Gamma \geq 0.5368$.

\begin{figure}
    \begin{center}
         \begin{tabular}{| c || c |c |c| c| c| c| c| c| c| c| c| c| c|}
         \hline
       $n$ & 0 & 1 & 2 & 3 & 4 & 5 & 6 & 7 & 8 & 9 &  $ n \geq 10$ \\
      $F(n)$  & 1.0 & 0.833 & 0.677 & 0.54 & 0.426 & 0.333 & 0.260 & 0.201 & 0.156 & 0.121 & $0.093 \cdot 0.773^{n - 10}$\\
      \hline
    \end{tabular}
    \end{center}
    \caption{We prove the existence of an $(F,6)$-discrete OCS for the $F$ described in this table. The construction is given in \Cref{subsec:construction-overview,sec:W} with the hyperparameter $p = 0.48$.}
    \label{fig:OCS-parameters}
\end{figure}

\section{From win distributions to discrete OCSes}
\label{sec:metathm-proof}
In this section, we prove the following meta theorem, restated for convenience. Recall the pseudocode for $\mathrm{OCS}_{\mathcal{W},\mathcal{T}}$ is given in \Cref{fig:OCS meta algorithm}.
\metaalg*

We break the proof of \Cref{thm:meta-algorithm} into two steps: In \Cref{lem:meta-NA} we assuming the existence
of a tournament subroutine with some desirable properties. Then, in \Cref{lem:small-to-T}, we prove the
existence of such a tournament subroutine as long as the win distribution is sufficiently small.

\subsection{Proof of \Cref{sec:metathm-proof} assuming a good tournament subroutine}
We prove the following Lemma.
\begin{lemma}
    \label{lem:meta-NA}
    Let $\mathcal{W}$ be a win-distribution with the following properties.
    \begin{enumerate}
        \item \textbf{Negatively associated:} For any $n \geq 1$ and $\bx \sim \W$, the variables $\bx_1, \ldots, \bx_n$ are negative associated.
        \item \textbf{Shift-invariant:} For any $n \geq 1$ and $\bx \sim \W$, the distribution of $\bx_1, \bx_2, \ldots$ is identical to that of $\bx_n, \bx_{n+1}, \ldots$.
    \end{enumerate}
    Let $\mathcal{T}$ be a tournament subroutine with the following properties.
    \begin{enumerate}
        \item \textbf{Consistent:} For each $w \in [0,1]$, the probability $i$ wins round $j$ given that $\bw_j(i) = w$ is at least $w$.
        \item \textbf{Stochastic ordering:} Whether $i$ wins round $j$ is a stochastically increasing function of $\bw_j(i)$ and a stochastically decreasing function of $\bw_j(i')$ for each $i' \neq i$.
    \end{enumerate}
    Then, $\mathrm{OCS}_{\mathcal{W}, \mathcal{T}}$ is a $(F,m)$-discrete OCS for
    \begin{align*}
        F(n) \coloneqq \Ex_{\bx \sim \D} \left[\prod_{\ell = 1}^n(1 - \bx_\ell) \right].
    \end{align*}
\end{lemma}

We'll need a variety of desirable properties satisfied by negatively associated variables. All of the below facts have short proofs (see~\cite{JP83}).
\begin{fact}
    \label{fact:indep-are-NA}
    If $\bx_1, \ldots, \bx_n$ are independent, then they are negatively associated.
\end{fact}
\begin{fact}
    \label{fact:union-NA}
    The union of independent sets of negatively associated random variables is also negatively associated.
\end{fact}
\begin{fact}
    \label{fact:sibset-NA}
    Any subset of a set of negatively associated variables is also negatively associated.
\end{fact}
\begin{fact}
    \label{fact:increasing-NA}
    Increasing functions defined on disjoint subsets of negatively associated random variables are also negatively associated.
\end{fact}

Next, we prove an analogue of \Cref{fact:increasing-NA} for stochastically increasing functions.
\begin{proposition}
    \label{prop:NA-closed-stochastically-increasing}
    Stochastically increasing functions defined on disjoint subsets of negatively associated random variables are also negatively associated.
\end{proposition}
\begin{proof}
    Let $\bx = \bx_1, \ldots, \bx_n$ be negatively associated, $A_1, \ldots, A_m \subseteq [n]$ disjoint subsets, and $\boldf_1, \ldots, \boldf_m$ stochastically increasing functions. We draw $\by_1, \ldots, \by_m \overset{iid}{\sim} \Uniform(0,1)$. By \Cref{fact:indep-are-NA} and \Cref{fact:union-NA}, the set $(\bx_1, \ldots, \bx_n, \by_1, \ldots, \by_m)$ is together negatively associated. Then, we define $\bz_i$ for each $i \in M$ as,
    \begin{align*}
        \bz_i = \text{min $v$ such that } \Pr[\boldf_i(\bx_{A_i}) \leq v] \geq \by_i
    \end{align*}
    Then the distribution of $\bz_i$ is equivalent to that of $\boldf_i(\bx_{A_i})$. Furthermore, since $\boldf_i$ is stochastically increasing, $\bz_i$ is an increasing function of $\bx_{A_i}$ and $\by_i$. By \Cref{fact:increasing-NA}, $\bz_1, \ldots, \bz_m$ are negatively associated.
\end{proof}
As a quick corollary of the above,
\begin{corollary}
    \label{cor:NA-stochastic-product}
    Let $\bx = \bx_1, \ldots, \bx_n$ be negatively associated. For any stochastically increasing functions $\boldf_1, \ldots, \boldf_m$ defined on disjoint coordinates of $\bx$, $A_1, \ldots, A_m$ respectively,
    \begin{align*}
        \Ex\left[\prod_{i = 1}^m \boldf_i(\bx_{A_i})\right] \leq \prod_{i = 1}^m \Ex\left[\boldf_i(\bx_{A_i})\right]
    \end{align*}
\end{corollary}
\begin{proof}
    Let $\bz_i$ be the output of $\boldf_i(\bx_{A_i})$ for $i = 1, \ldots, m$. Then, by \Cref{prop:NA-closed-stochastically-increasing}, $\bz_1, \ldots, \bz_m$ are negatively associated. By repeated application of \Cref{def:NA}, we have the desired result.
\end{proof}
We'll also need that negative association is preserved under negation.\footnote{Negative association is only preserved under negation of \emph{all} the variables. Negating only some variables can create a positive correlation.}
\begin{proposition}
    \label{prop:NA-negation}
    Let $\bx_1, \ldots, \bx_n$ be negatively associated. Then $-\bx_1, \ldots, -\bx_n$ are also negatively associated
\end{proposition}
\begin{proof}
    Let $f_1, f_2$ be increasing functions over disjoint coordinates $A_1, A_2 \subseteq [n]$. We wish to prove
    \begin{align*}
        \Cov[f_1((-\bx)_{A_1}), f_2((-\bx)_{A_2})] \leq 0.
    \end{align*}
    Since $-f_1((-\bx)_{A_1})$ and $-f_2((-\bx)_{A_2})$ are increasing and disjoint functions of $\bx$ and $\bx$ is negatively associated,
    \begin{align*}
        \Cov[-f_1((-\bx)_{A_1}), -f_2((-\bx)_{A_2})] \leq 0.
    \end{align*}
    The Lemma follows from the fact that $\Cov[-\ba, -\bb] = \Cov[\ba, \bb]$ for any random variables $\ba, \bb$.
\end{proof}
    
With the above properties established, we are able to make progress towards proving \Cref{lem:meta-NA}. First, we establish that the desired win probabilities are negatively associated.
\begin{proposition}
    \label{prop:w-NA}
    Under the conditions of \Cref{lem:meta-NA}, for each $i \in L$, the variables $\bw_1(i), \bw_2(i), \ldots$ in \Cref{fig:OCS meta algorithm} are negatively associated.
\end{proposition}
\begin{proof}
    First, we claim that for any $r \geq 1$, the function $1 - \prod_{i=1}^r (1 - x_i)$ is an increasing function of $x_1, \ldots, x_r$. This is by a series of compositions.
    \begin{align*}
        &(x_i \mapsto 1-x_i) &\text{is decreasing}\\
        \implies& ((x_1, \ldots, x_r) \mapsto \prod_{i=1}^r (1 - x_i)) &\text{is decreasing}\\
         \implies& ((x_1, \ldots, x_r) \mapsto 1 - \prod_{i=1}^r (1 - x_i)) &\text{is increasing}
    \end{align*}
    Therefore, $\bw_j(i)$ are increasing functions of disjoint subsets of $(\bx_1(i), \bx_2(i), \ldots)$. The desired result follows from \Cref{fact:increasing-NA}.
\end{proof}

\begin{lemma}
    \label{lem:lose-disjoint-sets}
    Under the conditions of \Cref{lem:meta-NA}, for any $i \in L$ and disjoint subsets of time steps $S_1, \ldots, S_k \subseteq [T]$
    \begin{align*}
        \Pr[i \text{ never a winner in } \bigcup_{\ell = 1}^k S_\ell] \leq \prod_{\ell = 1}^k \Pr[i \text{ never a winner in } S_\ell] 
    \end{align*}
\end{lemma}
\begin{proof}
    According to the properties of $\mathcal{T}$, whether $i$ never wins in $S_\ell$ is a stochastically increasing function of $(-\bw_j(i))$ and of $\bw_j(i')$ for each $i' \neq i, j\in S_\ell$. By \Cref{prop:NA-negation} and \Cref{prop:w-NA}, the variables $(-\bw_1(i), -\bw_2(i), \ldots)$ are negatively associated. Applying \Cref{fact:union-NA}, we have that
    \begin{align*}
        \bigcup_{j = 1}^\infty \left(\{-\bw_j(i)\} \cup \bigcup_{i' \neq i} \bw_j(i') \right)  \text{ are collectively negatively associated.}
    \end{align*}
     The desired results follows \Cref{cor:NA-stochastic-product}.
\end{proof}

\Cref{lem:lose-disjoint-sets} allows us to prove \Cref{lem:meta-NA} by just considering how the OCS behaves on consecutive subsequences rather than on unions of consecutive subsequences. We prove a final ingredient.

\begin{lemma}
    \label{lem:lose-seperate}
    Under the conditions of \Cref{lem:meta-NA}, for any $i \in L$, subset of time steps $S \subseteq [T]$ in which $i$ appears, and desired win probabilities $w_j(i)$ for each $j \in S$, 
    \begin{align*}
        \Pr[i \text{ never a winner in } S \mid \bw_j(i) = w_j(i) \text{ for each } j \in S] \leq \prod_{j \in S} (1 - w_j(i)).
    \end{align*}
\end{lemma}
\begin{proof}
    Just as in \Cref{lem:lose-disjoint-sets},
    \begin{align*}
         \bigcup_{j = 1}^\infty \left(\bigcup_{i' \neq i} \bw_j(i') \right)  \text{ are all negatively associated.}
    \end{align*}
    Whether $i$ loses in round $j$ is a stochastically increasing function of $w_j(i')$ for each $i \neq i'$. Applying \Cref{cor:NA-stochastic-product},
    \begin{align*}
        \Pr[i \text{ never a winner in } &S \mid \bw_j(i) = w_j(i) \text{ for each } j \in S] \\
        &\leq \prod_{j \in S} \left(\Pr[i \text{ loses in round } j \mid \bw_j(i) = w_j(i)]\right)
    \end{align*}
    By the consistency of $\mathcal{T}$, $\Pr[i \text{ loses in round } j \mid \bw_j(i) = w_j(i)]$ is at most $1 - w_j(i)$, proving the desired result.
\end{proof}

Finally, we prove \Cref{lem:meta-NA}.
\begin{proof}[Proof of \Cref{lem:meta-NA}]
    Our goal is to prove that for any vertex $i \in L$ and disjoint consecutive time steps $S_1, \ldots, S_k$, that
    \begin{align*}
        \Pr[i \text{ never a winner in }S_1, \ldots, S_k] \leq \prod_{\ell = 1}^k F(\mathrm{count}_i(S_\ell))
    \end{align*}
    where $f$ is as defined in \Cref{eq:def f}. By \Cref{lem:lose-disjoint-sets} and the assumption that $\mathcal{W}$ is shift-invariant, it is enough to prove, for $r_j(i) \coloneqq$ the number of appearances of $i$ in $A_j$,
    \begin{align*}
        \Pr[i \text { never a winner in $S \coloneqq \{1, \ldots, m\}$}] \leq F\left(\sum_{j = 1}^m r_j(i)\right).
    \end{align*}
    We compute,
    \begin{align*}
        \Pr[i \text { never a winner in $S$}] &= \Ex_{\bw_1(i), \ldots, \bw_m(i)}\left[\Pr[i \text { never a winner in $S$} \mid \bw_1(i), \ldots, \bw_m(i) ]  \right]\\
        & \leq \Ex_{\bw_1(i), \ldots, \bw_m(i)}\left[\prod_{j \in S}(1 - \bw_j(i))  \right] \tag{\Cref{lem:lose-seperate}}\\
        & = \Ex_{\bx(i) \sim \W}\left[\prod_{j \in S}\left( \prod_{\ell = 0}^{r_j(i)-1} \left(1 - \bx_{k_j(i) + \ell}\right)\right) \right]\tag{\Cref{eq:desired-win-prob-prod}} \\
         & = \Ex_{\bx(i) \sim \W}\left[\prod_{\ell = 1}^{r_1(i) + \ldots + r_m(i)}\left( 1 - \bx_{\ell}\right) \right]\tag{Combine products} \\
         &= F\left(\sum_{j \in S} r_j(i)\right) \tag{\Cref{eq:def f}}
    \end{align*}
    as desired.
\end{proof}

\subsection{Design of the tournament subroutine}

In this subsection, we prove the following Lemma which completes the proof of \Cref{thm:meta-algorithm}.
\begin{lemma}[Sufficiently small $\implies$ good $\mathcal{T}$]
    \label{lem:small-to-T}
    Let $\mathcal{W}$ be a win-distribution for which the following holds. For each $r \in [m-1]$ define $\bw^{(r)}$ to be the random variable representing the distribution of the desired win probability of a vertex of multiplicity $r$,
    \begin{align*}
       \bw^{(r)} \sim \left(1 - \prod_{\ell = 1}^{r} \left(1 - \bx_{\ell}\right) \right) \quad\quad \text{where} \quad\quad \bx \sim \W,
    \end{align*}
    If, for each $t \in [0,1]$ and $r \in [m-1]$,
    \begin{align}
        \label{eq:w-small}
        \Ex_{\bx \sim \W} \left[\left(\bw^{(r)} - t \right)_+ \right] \leq 1 - t  - \frac{m-r}{m} \cdot \left (1 - t^{\frac{m}{m-r}} \right),
    \end{align}
    then there is a tournament subroutine, $\mathcal{T}$ with the following properties.
    \begin{enumerate}
        \item \textbf{Consistent:} For each $w \in [0,1]$, the probability $i$ wins round $j$ given that $\bw_j(i) = w$ is at least $w$.
        \item \textbf{Stochastic ordering:} The probability $i$ wins a round is a stochastically increasing function of $\bw_j(i)$ and a stochastically decreasing function of $\bw_j(i')$ for each $i' \neq i$.
    \end{enumerate}
\end{lemma}
Our tournament subroutine will be simple: We'll design stochastically increasing strength functions $\BS_1, \ldots, \BS_{m-1}$  that map a desired win probability to a ``strength." For each $i \in A_j$, we set its strength to $\BS_{r_j(i)}(\bw_j(i))$, and the vertex with the highest strength wins the round. We give pseudocode for this subroutine in \Cref{fig:tournament}.

\begin{figure}[h]
  \captionsetup{width=.9\linewidth}
\begin{tcolorbox}[colback = white,arc=1mm, boxrule=0.25mm]
\vspace{3pt} 
$\mathrm{\mathcal{T}}_{\BS_1, \ldots, \BS_{m-1}}$
\begin{enumerate}
    \item If there is a vertex with $r_j(i) = m$, then select $i$ as the winner.
    \item Otherwise, for each $i \in A_j$, sample a strength,
    \begin{align*}
        \bs_j(i) \coloneqq \BS_{r_j(i)}(\bw_j(i)),
    \end{align*}
    and choose the $i$ maximizing $\bs_j(i)$ as the winner.
\end{enumerate}
\end{tcolorbox}
\caption{A tournament subroutine for returning a winner given desired win probabilities $\bs_j(i)$ and multiplicities $r_j(i)$. This subroutine is parameterized by strength functions $\BS_1, \ldots, \BS_{m-1}$.}
\label{fig:tournament}
\end{figure}

As long as the strength functions are stochastically increasing, the stochastic ordering criteria of \Cref{lem:small-to-T} will hold. On the other hand, consistency is more challenging. The following fundamental question underlies this subsection.

\begin{quote} {\sl For which distributions of $\bw^{(r)}$ is it possible to design stochastically increasing strength functions such that the probability $i$ wins round $j$ given $\bw_j(i) = w$ is at least $w$.} \hfill{($\diamondsuit$)} 
\end{quote} 
Consider the case where $r = 1$. Since there can only be one winner in a round, for ($\diamondsuit$) to hold, it must be the case that $\E[\bw^{(1)}] \leq \frac{1}{m}$. It turns out, that condition is not sufficient. Suppose that $\bw^{(1)}$ is $1$ with probability $\frac{1}{m}$ and $0$ otherwise. Then, there is a nonzero probability that $2$ (or more) vertices will both declare that they should win with probability $1$, making ($\diamondsuit$) impossible. This example shows that not only do we want the expectation of $\bw^{(1)}$ to be small, we want its distribution to not be too spread out.

We will show that these two notions -- how large and spread out $\bw^{(r)}$ is -- fully capture when \hfill{($\diamondsuit$)}  is possible. Recall the following definition which formalizes the notion of ``large and spread out."
\icxdef*

The following Lemma formalizes our intuition that as long as the desired win probabilities are sufficiently small and not spread out, ($\diamondsuit$) holds.
\begin{lemma}[Increasing convex order $\to$ ($\diamondsuit$)]
    \label{lem:icx-to-strength-dist}
    For any random variables $\bw^{(1)} \ldots, \bw^{(m-1)}$ and $\bw'^{(1)} \ldots, \bw'^{(m-1)}$, if
    \begin{align*}
        \bw^{(r)} \icx \bw'^{(r)} \quad \quad \text{for each $r = 1, \ldots, m-1$},
    \end{align*}
    and ($\diamondsuit$) holds for $(\bw'^{(1)}, \ldots, \bw'^{(m-1)})$, then it also holds for $(\bw^{(1)}, \ldots, \bw^{(m-1)})$.
\end{lemma}

Our proof of \Cref{lem:icx-to-strength-dist} leverages a powerful theorem of M{\"u}ller and R{\"u}schendorf.
\begin{theorem}[\cite{MR01}]
    \label{thm:mr01}
    For univariate random variables $\ba, \bb$, the following are equivalent.
    \begin{enumerate}
        \item $\ba$ is smaller than $\bb$ in the increasing convex order.
        \item There is a coupling of $(\hat{\ba}, \hat{\bb})$ such that the marginal distribution of $\hat{\ba}$ and $\hat{\bb}$ are the same as those of $\ba$ and $\bb$ respectively, and, for any choice of $a$,
        \begin{align*}
            \Ex[\hat{\bb} \mid \hat{\ba}] \geq \hat{\ba}.
        \end{align*}
        Furthermore, the stochastic function mapping $a \mapsto (\hat{\bb} \mid \hat{\ba} = a)$ is stochastically increasing.
    \end{enumerate}
\end{theorem}
\Cref{thm:mr01} is a generalization of a classic result by Strassen \cite{s65} who proved the same result without the restriction that $a \mapsto (\hat{\bb} \mid \hat{\ba} = a)$ be stochastically increasing. We'll need the stronger result since we want our strength functions to be stochastically increasing.
\begin{proof}[Proof of \Cref{lem:icx-to-strength-dist}]
    For each $r \in [m-1]$, let $(\hat{\bw}^{(r)}, \hat{\bw}'^{(r)})$ be the coupling guaranteed to exist by \Cref{thm:mr01} and let $\BS_1', \ldots,\BS_{m-1}'$ be the strength functions satisfying ($\diamondsuit$) for $\bw'^{(1)} \ldots, \bw'^{(m-1)}$. Define, for each $r \in [m-1]$,
    \begin{align*}
        \BS_r(w) \coloneqq \BS_r'(\hat{\bw}') \quad \text{where} \quad\hat{\bw'} \sim ( \hat{\bw}'^{(r)} \mid \hat{\bw}^{(r)} = w).
    \end{align*}
    We claim that $\BS_1, \ldots, \BS_{m-1}$ are strength functions satisfying ($\diamondsuit$) for $(\bw^{(1)}, \ldots, \bw^{(m-1)})$. First, since the composition of stochastically increasing functions is stochastically increasing, $\BS_r$ is stochastically increasing for all $r \in [m-1]$. Second, the probability $i$ wins a round with strength $w$ is at least $\E[\hat{\bw}'^{(r)} \mid \hat{\bw}^{(r)} = w]$, which is at least $w$.
\end{proof}

\begin{remark}[Computational efficiency of \Cref{lem:icx-to-strength-dist}] Computing $\BS_r$ requires knowing the coupling $(\hat{\bw}^{(r)}, \hat{\bw}'^{(r)})$. If $\hat{\bw}^{(r)}$ and $ \hat{\bw}'^{(r)}$ are discrete random variables, finding a desired coupling is simply a linear program. If they are continuous, they can be discretized to any desired accuracy and then treated as discrete.
\end{remark}

In order to use \Cref{lem:icx-to-strength-dist}, we need to be able to verify when one variable is smaller, in the increasing convex order, than another. We use the following well known characterization.
\begin{fact}[\cite{SS07}]
    \label{fact:icx-expectations}
    For univariate random variables $\ba, \bb$, the following are equivalent.
    \begin{enumerate}
     \item $\ba$ is smaller than $\bb$ in the increasing convex order.
     \item For every $t \in \R$,
     \begin{align*}
         \Ex[(\ba - t)_+] \leq \Ex[(\bb - t)_+] .
     \end{align*}
     \end{enumerate}
\end{fact}

We complete the proof of \Cref{lem:small-to-T} by giving a particular choice for $(\bw^{(1)}, \ldots, \bw^{(m-1)})$ on which ($\diamondsuit$) holds.
\begin{lemma}
    \label{lem:large-diamond-holds}
    For each $r  \in [m-1]$, let $\bw^{(r)}$ be the random variable supported on $[0,1]$ satisfying
    \begin{align}
        \label{eq:good-w-expectation}
        \Ex[(\bw^{(r)} - t)_+] = 1 - t  - \frac{m-r}{m} \cdot \left (1 - t^{\frac{m}{m-r}}\right) \quad \quad\forall t\in [0,1]
    \end{align}
    Then, ($\diamondsuit$) holds for $(\bw^{(1)}, \ldots, \bw^{(m-1)})$
\end{lemma}
\begin{proof}
    Straightforward computation verifies that \Cref{eq:good-w-expectation} is equivalent to each $\bw^{(r)}$ having the CDF.
    \begin{align*}
        \Pr[\bw^{(r)} \leq t] = t^{\frac{r}{m-r}} \quad \quad \forall t \in [0,1].
    \end{align*}
    Consider the strength functions, for each $r \in [m-1]$,
    \begin{align*}
        \BS_r(w) \coloneqq w^{\frac{1}{m-r}}.
    \end{align*}
    For each $r \in [m-1]$, we have that
    \begin{align*}
        \Pr[\BS_r(\bw^{(r)}) \leq t] = \Pr[\bw^{(r)} \leq t^{m-r}] = t^r.
    \end{align*}
    Therefore, $\BS_r(\bw^{(r)})$ is distributed according to the max of $r$ independent variables drawn from $\Uniform(0,1)$. A vertex $i$ is selected in a round if its strength is larger than the strengths of each of the other vertices. If $i$ has multiplicity $r(i)$, then the max strength of the other vertices is distributed according to the maximum of $m - r(i)$ independent uniforms. Therefore,
    \begin{align*}
        \Pr[i\text{ selected with strength }w] &= \Pr[\text{max of ($m - r(i)$) independent uniforms} < \BS_{r(i)}(w)] \\
        &= \BS_{r(i)}(w)^{m - r(i)}\\
        &= w,
    \end{align*}
    Furthermore, $\BS_r$ is increasing for all $r  \in [m-1]$. Therefore, both criteria of ($\diamondsuit$) hold.
\end{proof}

\section{Constructing a win distribution}
\label{sec:W}

In this section we will construct a distribution $(\bx_1, \bx_2, \ldots) \sim \W$ that meets the three criteria of \Cref{thm:meta-algorithm} and also leads to a good competitive ratio. We do this in two steps. First, we construct a ``seed distribution" $\D$ that returns an infinite sequence $(\by_1, \by_2, \ldots)$ of random variables that are negatively associated and shift-invariant, but the values $\by_j$ do not represent probabilities (they are unbounded). Then, for a carefully chosen increasing function $W$, we set $\bx_k = W(\by_k)$ to ensure the third criteria of \Cref{thm:meta-algorithm} holds. We can think of $W$ as mapping an output of $\D$ to a desired win probability.

In the remainder of this section, $p$ is a hyperparameter than impacts the distribution of $\D$ and $\W$. For concreteness, we will eventually set $p = 0.48$ to optimize the competitive ratio, a value we selected using brute-force search.

\begin{definition}[$\D$]
    \label{def:D}
    A sample $(\by_1, \by_2, \ldots) \sim \D$ is generated by the following process, where each random choice is independent.
    \begin{enumerate}
        \item Initialize $\bz_1 \sim \Geo(p)$.
        \item For $k = 1,2, \ldots$,
        \begin{enumerate}
            \item With probability $p$, set $\by_k = \bz_k$ and then reset $z$, setting $\bz_{k+1} \leftarrow 0$.
            \item Otherwise, set $\by_k=-1$ and increment $z$, setting $\bz_{k+1} \leftarrow \bz_{k}+ 1$
        \end{enumerate}
    \end{enumerate}
\end{definition}

It will often by convenient for us to refer to the following variants of $\D$.
\begin{definition}
    We use the following notation to specify variants of $\D$.
    \begin{enumerate}
        \item For any $v \geq 0$, we use $\D_v$ to refer to the generative process of \Cref{def:D} where $\bz_1$ is initialized to $v$ instead of drawn from $\Geo(p)$.
        \item For any $n \geq 1$, we use $\Dn$ to refer to first $n$ element of $\D$; i.e. to $\by_1, \ldots, \by_n$ where $\by \sim D$.
        \item We use $\D(y,z)$ to refer to the distribution that outputs two infinite sequences, one corresponding to the $(\by_1, \by_2, \ldots)$ and the other to the $(\bz_1, \bz_2, \ldots)$ in \Cref{def:D}.
    \end{enumerate}
    We'll also mix and match the above notation. For example, we might use $\by, \bz \sim \Dn(y,z)$ as shorthand for $\by$ and $\bz$ are each $n$-tuples containing $\by_1, \ldots, \by_n$ and $\bz_1, \ldots, \bz_n$ respectively.
\end{definition}
By examining the generative process in \Cref{def:D}, we observe that $\D$ is \emph{almost} memoryless.
\begin{fact}[$\D$ is memoryless except for $z$]
    \label{fact:memoryless}
    For any $k, v \geq 0$ and $\by, \bz \sim \D(y,z)$, the following distributions are identical.
    \begin{enumerate}
        \item The distribution of $(\by_k, \by_{k+1}, \ldots)$ conditioned on $\bz_k = v$.
        \item The distribution $\D_{v}$.
    \end{enumerate}
\end{fact}

First, we prove that $\D$ is shift-invariant. This amounts to proving that for any $k \geq 1$, $\bz_k$ is distributed according to a $\Geo(p)$.
\begin{proposition}
    \label{prop:z-geo}
    Let $\by, \bz \sim \D(y,z)$. Then, for any $k \geq 1$, $\bz_k$ is distributed according to a $\Geo(p)$ distribution,
    \begin{align*}
        \Pr[\bz_k = v] = (1 - p)^v \dot p.
    \end{align*}
\end{proposition}
\begin{proof}
    By induction on $k$; $\bz_1$ is initialized from a geometric distribution, so \Cref{prop:z-geo} holds for $k = 1$.
    
    For the inductive step, consider some $k > 1$, $\bz_k = 0$ only it was reset at that round, which happens with probability $p$. For $v \geq 1$
    \begin{align*}
        \Pr[\bz_k = v] &= (1 - p) \cdot \Pr[\bz_{k-1} = v-1] \\
        &= (1 - p) \cdot (1 - p)^{v-1} \cdot p\\
        &= (1-p)^v \cdot p.
    \end{align*}
    Therefore, $\bz_k \sim \Geo(p)$.
\end{proof}
As a corollary of \Cref{fact:memoryless} and \Cref{prop:z-geo}, the following holds.
\begin{corollary}[$\D$ is shift-invariant]
   For any $k \geq 1$ and $\by \sim \D$, the distribution of $\by_1, \by_2, \ldots$ is identical to that of $\by_k, \by_{k+1}, \ldots$.
\end{corollary}

\subsection{$\D$ is negatively associated}
We devote this subsection to establishing that $\D$ is negatively associated. Recall the definition for negatively associated variables.
\NA*

For starters, we prove (a strengthening of) the above definition holds for $\Dn$ and $\Dn_0$ if we restrict $A_1 = \{n\}$.
\begin{lemma}
    \label{lem:NA-just-last-bit}
    For any $n \geq 1$ and indices $S \subseteq [n - 1]$, let $\by_1, \ldots, \by_n$ be drawn as
    \begin{align*}
        (\by_1, \ldots, \by_n) \sim (\Dn \mid \by_i = -1 \text{ for each $i \in S$})
    \end{align*}
   Then, for any increasing functions $f_1: \Z \to \R, f_2: \Z^{n-1} \to \R$,
    \begin{align*}
        \Covx_{\by_1, \ldots, \by_n }\big[f_1(\by_n), f_2(\by_1, \ldots, \by_{n-1} )\big] \leq 0
    \end{align*}
    The same is true when using $\Dn_0$ rather than $\Dn$.
\end{lemma}
We only prove \Cref{lem:NA-just-last-bit} for $\Dn$ as the proof for $\Dn_0$ is identical. Throughout the following proof, the distribution on $\by_1, \ldots, \by_n$ is by default $(\Dn \mid \by_i = -1 \text{ for each $i \in S$})$.
\begin{proof}
    Without loss of generality, we can assume $f_1(-1) = 0$; otherwise, add an appropriate constant to $f_1$ so that it's true. Since $f_1$ is an increasing function of $\by_n$ and $\by_n \geq -1$, we can write,
    \begin{align*}
        f_1(k) = \sum_{j=0}^k g(j)
    \end{align*}
    for an appropriate nonnegative function $g: \Zpos \to \R_{\geq 0}$. We then expand the first term of the desired covariance.
    \begin{align*}
        \Ex_{\by_1, \ldots, \by_n }\big[f_1(\by_n) \cdot  f_2(\by_1 \cdot \ldots, \by_{n-1} )\big] &= \sum_{k = -1}^{\infty} f_1(k)  \cdot  \Pr[\by_n = k] \cdot \Ex\big[f_2(\by_1, \ldots, \by_{n-1}) \mid \by_n = k\big] \\
        &= \sum_{k = -1}^{\infty} \sum_{j=0}^k g(j) \cdot \Pr[\by_n = k] \cdot  \Ex\big[f_2(\by_1, \ldots, \by_{n-1}) \mid \by_n = k\big] \\
        &= \sum_{j = 0}^{\infty} \sum_{k = j}^{\infty} g(j) \cdot \Pr[\by_n = k] \cdot  \Ex\big[f_2(\by_1, \ldots, \by_{n-1}) \mid \by_n = k\big] \\
        &= \sum_{j = 0}^{\infty} g(j) \cdot \Pr[\by_n \geq j] \cdot  \Ex\big[f_2(\by_1, \ldots, \by_{n-1}) \mid \by_n \geq j\big]
    \end{align*}
    
    We claim that $\Ex\big[f_2(\by_1, \ldots, \by_{n-1}) \mid \by_n \geq j\big]$ is a decreasing function of $j$. First we note that conditioning on $(\by_n \geq j)$ forces that $\by_{n - j} =\ldots = \by_{n - j - 1} = - 1$ and leaves the distribution of $\by_1, \ldots, \by_{n - j -2}$ unchanged.
    
    Suppose we draw $\by_1, \ldots, \by_{n-1}$ conditioned on $(\by_n \geq j)$. Then, we set $\by_{n - (j + 1)}$ to $-1$ and leave all other $\by_1, \ldots, \by_{n-1}$ unchanged. This is equivalent to drawing $\by_1, \ldots, \by_{n-1}$ conditioned on $(\by_n \geq j + 1)$. Importantly, setting $\by_{n - (j + 1)} = -1$ can only decrease the value of $f_2(\by_1, \ldots, \by_{n-1})$, because $f_2$ is increasing. Therefore, we have that,
    \begin{align}
        \label{eq:conditionally-decreasing}
        &\Ex\big[f_2(\by_1, \ldots, \by_{n-1}) \mid \by_n \geq j\big] \quad \quad\text{is decreasing in $j$}.
    \end{align}
    We can now prove the desired result.
    \begin{align*}
        \Ex\big[&f_1(\by_n) \cdot  f_2(\by_1 \cdot \ldots, \by_{n-1} )\big] \\
        &= \sum_{j = 0}^{\infty} g(j) \cdot \Pr[\by_n \geq j] \cdot  \Ex\big[f_2(\by_1, \ldots, \by_{n-1}) \mid \by_n \geq j\big] \\
        &\leq  \sum_{j = 0}^{\infty} g(j) \cdot \Pr[\by_n \geq j] \cdot  \Ex\big[f_2(\by_1, \ldots, \by_{n-1}) \mid \by_n \geq -1\big] \\
        &= \left(\sum_{j = 0}^{\infty} g(j) \cdot \Pr[\by_n \geq j]\right) \cdot  \left(\Ex\big[f_2(\by_1, \ldots, \by_{n-1})\big]\right) \\
        &= \Ex\big[f_1(\by_n)\big] \cdot \Ex\big[f_2(\by_1 \cdot \ldots, \by_{n-1} )\big].
    \end{align*}
    Where the inequality follows from \Cref{eq:conditionally-decreasing}.
\end{proof}

Next, we prove that $\Dn$ is negatively associated.
\begin{lemma}
    \label{lem:DN-NA}
    For any $n \geq 1$, draw $\by_1, \ldots, \by_n \sim \Dn$. Then, $\by_1, \ldots, \by_n$ are negatively associated.
\end{lemma}
\begin{proof}
    By induction on $n$. For $n = 1$, the desired result holds because a single variable is vacuously negatively associated.
    
    For any $n \geq 2$ fix some disjoint sets $A_1, A_2 \subseteq [n]$ and increasing functions $f_1: \Z^{|A_1|} \to \R, f_2:\Z^{|A_2|} \to \R$. Without loss of generality, we can assume that $n \in A_1$ and $n \notin A_2$. We define an ``averaging" of $f_1$ that does \emph{not} depend upon $\by_n$.
    \begin{align*}
        \bar{f}_1(\by_{A_1}) \coloneqq \Ex_{\by \sim \Dn} \big[f_1(\by_{A_1}) \mid \by_{i} = \by_i \text{ for each $i \in (A_1 \setminus \{n\})$} \big]
    \end{align*}
    Importantly, neither $\bar{f_1}$ nor $f_2$  depend on $\by_n$, so we can apply the inductive hypothesis.
    \begin{align}
        \label{eq:f-bar-inductive}
        \Ex[\bar{f}_1(\by_{A_1}) \cdot f_2(\by_{A_2})] \leq  \Ex[\bar{f}_1(\by_{A_1})] \cdot \Ex[f_2(\by_{A_2})]
    \end{align}
    Furthermore, $\bar{f}_1$ and $f_1$ have the same expectation. Our next goal is to show that
    \begin{align*}
        \Ex[f_1(\by_{A_1}) \cdot f_2(\by_{A_2})] \leq \Ex[\bar{f}_1(\by_{A_1}) \cdot f_2(\by_{A_2})]
    \end{align*}
    The above two equations will imply the desired result. 
    
    Define the ``left signature" of some $y_1, \ldots, y_n$ as follows. If $y_i = -1$ for all $i \in A_1 \setminus \{n\}$, then the signature is $\mathcal{L}(y_1, \ldots, y_n) = \emptyset$. Otherwise, let $i^\star$ be the largest $i \in A_1 \setminus \{n\}$ where $y_i \geq 0$ and the signature is $\mathcal{L}(y_1, \ldots, y_n) = (y_1, \ldots, y_{i^\star})$. We claim for any choice of left signature, $\ell$, that
    \begin{align*}
        \Ex[f_1(\by_{A_1}) \cdot f_2(\by_{A_2}) \mid \mathcal{L}(\by) = \ell] \leq \Ex[\bar{f}_1(\by_{A_1}) \cdot f_2(\by_{A_2}) \mid \mathcal{L}(\by) = \ell] .
    \end{align*}
    First consider the case where $\ell \neq \emptyset$. Conditioning upon signature $\ell$ is equivalent to:
    \begin{enumerate}
        \item Fixing the values of $\by_1, \ldots, \by_{i^{\star}}$ to those in $\ell$.
        \item Conditioning on $\by_j = -1$ for all $j \in A_1 \setminus \{n\}$ where $j > i^{\star}$. 
    \end{enumerate}
    By \Cref{fact:memoryless}, the distribution of $\by_{i^\star + 1}, \ldots, \by_n$ is the same as $\mathcal{D}^{(n - i^\star)}_0$ conditioned on some of the elements being $-1$. This is a distribution for which we can apply \Cref{lem:NA-just-last-bit}. We also need $f_1$ to just depend on the last bit, $\by_n$. Fortunately, conditioning on the signature fixes all bits in $A_1$ except for $\by_n$ so that property is satisfied. Therefore, by \Cref{lem:NA-just-last-bit}.
    \begin{align}
        \label{eq:cov-neg-conditioning-signature}
        \Ex[f_1(\by_{A_1}) \cdot f_2(\by_{A_2}) \mid \mathcal{L}(\by) = \ell] &\leq \Ex[f_1(\by_{A_1}) \mid \mathcal{L}(\by) = \ell] \cdot \Ex[f_2(\by_{A_2}) \mid \mathcal{L}(\by) = \ell]
    \end{align}
    In the other case where $\ell = \emptyset$, conditioning on $\ell$ is equivalent to drawing $\by_1, \ldots, \by_n$ from $(\Dn \mid y_i = -1 \text{ for each }i \in A_{1} \setminus \{n\})$. Once again, after this conditioning, $f_1$ only depends on $\by_n$, so we can apply \Cref{lem:NA-just-last-bit} and \Cref{eq:cov-neg-conditioning-signature} holds.
    
    Next, we wish to substitute in $\bar{f}_1$ for $f_1$. Fixing $\mathcal{L}(\by) = \ell$ fixes all variables in $A_1$ except for $n$. By the definition of $\bar{f}_1$, we have that
    \begin{align}
        \label{eq:f-bar-constant}
        \Ex[f_1(\by_{A_1}) \mid \mathcal{L}(\by) = \ell] = \bar{f}_1(x) \text{ for any } x \text{ where } \mathcal{L}(x) = \ell.
    \end{align}
    We'll use two aspects of the above equation: 
    \begin{enumerate}
        \item That $\bar{f}_1(x)$ has the same expectation as $f_1$ when conditioning on the signature.
        \item That $\bar{f}_1(x)$ is constant when conditioning on the signature.
    \end{enumerate},
    Combining the above with \Cref{eq:cov-neg-conditioning-signature}, where $\bell$ is drawn from $\mathcal{L}(\by)$ for $\by \sim \Dn$,
    \begin{align*}
        \Ex\big[f_1(\by_{A_1}) \cdot f_2(\by_{A_2})\big] &= \Ex_{\bell}\big[\Ex[f_1(\by_{A_1}) \cdot f_2(\by_{A_2}) \mid \mathcal{L}(\by) = \ell]\big] \tag{Law of total expectation}\\
        &\leq \Ex_{\bell}\big[\Ex[f_1(\by_{A_1}) \mid \mathcal{L}(\by) = \ell] \cdot \Ex[f_2(\by_{A_2}) \mid \mathcal{L}(\by) = \ell] \big] \tag{\Cref{eq:cov-neg-conditioning-signature}} \\
        &=  \Ex_{\bell}\big[\Ex[\bar{f}_1(\by_{A_1}) \mid \mathcal{L}(\by) = \ell] \cdot \Ex[f_2(\by_{A_2}) \mid \mathcal{L}(\by) = \ell] \big] \tag{\Cref{eq:f-bar-constant}}\\
        &=  \Ex_{\bell}\big[\Ex[\bar{f}_1(\by_{A_1}) \cdot f_2(\by_{A_2}) \mid \mathcal{L}(\by) = \ell] \big] \tag{$\bar{f}_1$ constant conditioned on $\ell$} \\
        &=  \Ex\big[\bar{f}_1(\by_{A_1}) \cdot f_2(\by_{A_2}) \big] \tag{Law of total expectation} \\
        &\leq \Ex\big[\bar{f}_1(\by_{A_1})\big] \cdot \Ex\big[f_2(\by_{A_2})\big] \tag{\Cref{eq:f-bar-inductive}} \\
        &= \Ex\big[f_1(\by_{A_1})\big] \cdot \Ex\big[f_2(\by_{A_2})\big] \tag{$\Ex\big[f_1\big] = \Ex\big[\bar{f}_1\big]$}
    \end{align*}
    Hence $f_1(\by_{A_1})$ and $f_2(\by_{A_2})$ are negatively correlated, so $\by_1, \ldots, \by_n$ satisfy the definition of negative association.
\end{proof}

\subsection{Choosing $W$}
Recall that we will set the output of $\W$ to $(W(\by_1), W(\by_2), \ldots)$ for $\by \sim \D$, where $W$ is some increasing function. Our goal in this section is to make $W$ as large as possible while still ensuring that the output of $\W$ isn't too large, meaning \Cref{eq:w-not-too-large} holds.

First, consider the case when every $A_j$ does not contain duplicates (meaning $r_j(i) \leq 1$ for all $i$). In this case, we have that the desired win probability of each vertex $i \in A_j$ is
\begin{align*}
    \bw_j(i) = \bx_{k_j(i)} = W(\by_{k_j(i)})
\end{align*}
One natural way to pick a winner this round is to pick the vertex with the largest $\by_{k_j(i)}$, breaking ties at random. Indeed, we choose $W$ in such a way that this occurs. Let $\bi_{\max}(y_1, \ldots, y_m)$ be the function that return the index $i$ maximizing $y_i$, choosing randomly from maximal $y_i$ in the case of a tie. We define
\begin{align}
    \label{eq:def-W}
    W(y) \coloneqq \Prx_{\by_2, \ldots, \by_m\overset{\mathrm{iid}}{\sim} \D^{(1)}}[\bi_{\max}(y, \by_2, \ldots, \by_m) = 1]
\end{align}
This $W(y)$ is a general choice than can work for any seed distribution $\D$ and, as long as $\D$ has negative correlations, this choice of $W$ results in $\W$ having negative correlations as well. Furthermore, the reason that we need \Cref{eq:w-not-too-large} is so that it's possible to design a tournament subroutine, and this choice of $W(y)$ allows for an easy tournament subroutine when each vertex only appears once in $A_j$: Just pick the vertex with largest desired win probability. Indeed, in \Cref{lem:suff-small-1}, we show that \Cref{eq:w-not-too-large} holds in the case of $r = 1$, though it still needs to be verified separately for $r \geq 2$. First, we compute the function $W$ for our $\D$.
\begin{lemma}
    \label{lem:compute-w}
    For $W$ defined in \Cref{eq:def-W} and $\D$ in \Cref{def:D},
    \begin{align}
        \label{eq:w-and-g}
        W(y) = \frac{G(y+1)^m - G(y)^m}{m\cdot(G(y+1) - G(y))} \quad\quad \text{where} \quad\quad G(v) \coloneqq \begin{cases}
        0 & \text{if } v = -1 \\
        (1 - p) + p \cdot (1 - (1 - p)^{v}) & \text{if } v \geq 0
        \end{cases}
    \end{align}
\end{lemma}
\begin{proof}
    Let $\by_1, \ldots, \by_m$ be $m$ independent samples from $\D^{(1)}$, and let $\bi^\star = \bi_{\max}(\by_1, \by_2, \ldots, \by_m)$. We wish to compute the probability $\bi^\star = 1$ as a function of $\by_1$. First, we observe that $G$ as defined in \Cref{eq:w-and-g} satisfies, for each $i = 1, \ldots, m$,
    \begin{align}
        \label{eq:G-property}
        G(v) = \Pr[\by_i \leq v - 1]
    \end{align}
    We consider a different process for generating $\by_1, \ldots, \by_m$ that simplifies the analysis. In general, we can generate any univariate random variable $\by$ by first generating $\ba \sim \Uniform(0,1)$ and then setting $\by$ to the smallest $v$ such that $\Pr[\by \leq v] \geq \ba$. Applying this to our setting, we first sample $\ba_1, \ldots, \ba_m \sim \Uniform(0,1)$ and then for each $i \in [m]$ set $\by_i$ to the smallest $v$ such that $G(v + 1) \geq \ba_i$. Then, we set $\bi^\star$ to the unique $i$ maximizing $\ba_i$ (which exists with probability $1$ since the $\ba_i$ are continuous). 
    
    This generating process for $(\by_1, \ldots, \by_m, \bi^\star)$ gives the same joint distribution as the original, so we are free to analyze it instead. Then, the distribution of $\ba_i$ conditioned on $\by_i$ is given by
    \begin{align*}
        (\ba_i \mid \by_i = y) \sim \Uniform(G(y), G(y + 1)).
    \end{align*}
    We next compute the desired result.
    \begin{align*}
        \Pr[\bi^\star = 1 \mid \by_1 = y] &= \Pr[\ba_{i} < \ba_1 \text{ for all $i = 2, \ldots, m$} \mid \by_1 = y] \\
        &= \Ex_{\ba_1} \left[(\ba_1)^{m - 1} \right] \tag{$\ba_{i} \sim \Uniform(0,1)$ for each $i  = 2, \ldots, m$} \\
        &= \frac{1}{G(y+1) - G(y)} \cdot \int_{G(y)}^{G(y+1)}
        t^{m-1}dt \tag{$\ba_1 \sim \Uniform(G(y),G(y+1))$ } \\
         &= \frac{G(y+1)^m - G(y)^m}{m\cdot(G(y+1) - G(y))} ,
    \end{align*}
    which is exactly the definition of $W(y)$.
    
\end{proof}

Next, we prove that using the $W$ from \Cref{eq:def-W} makes our win distribution satisfy \Cref{eq:w-not-too-large} for $r = 1$. We prove the following Lemma for our particular choice of $\D$ but note that a similar result holds for any $\D$ (as long as $W$ is defined as \Cref{eq:def-W} as a function of that $\D$).

\begin{lemma}[Sufficiently small for $r = 1$]
    \label{lem:suff-small-1}
    Let $W$ be as defined in \Cref{eq:def-W}. For any $t \in (0,1)$, 
    \begin{align*}
        \Ex_{\by \sim \D^{(1)}} \left[\left(W(\by)- t \right)_+ \right] \leq 1 - t  - \frac{m-1}{m} \cdot \left (1 - t^{\frac{m}{m-1}} \right)
    \end{align*}
    In particular, \Cref{eq:w-not-too-large} holds for $r = 1$.
\end{lemma}
\begin{proof}
    Straightforward computation verifies that if $\ba \sim \Uniform(0,1)$ and $\bz = \ba^{m-1}$, then for all $t \in (0,1)$
    \begin{align*}
        \Ex_{\bz} \left[\left(\bz- t \right)_+ \right] = 1 - t  - \frac{m-1}{m} \cdot \left (1 - t^{\frac{m}{m-1}} \right).
    \end{align*}
    Let $\bw$ have the distribution of $W(\by)$ when $\by \sim \D^{(1)}$. By \Cref{fact:icx-expectations}, the desired result is equivalent to showing that $\bw \icx \bz$. We'll instead prove $\bw \cx \bz$ (a stronger statement) and do so via \Cref{fact:cx-coupling}.
    
    To use \Cref{fact:cx-coupling}, we need to couple $\bw$ and $\bz$. Consider the following coupled generating process.
    \begin{enumerate}
        \item Draw $\ba \sim \Uniform(0,1)$.
        \item Set $\hat{\bz} = \ba^{m-1}$.
        \item Set $\hat{\by}$ to smallest $v \in \Z_{\geq -1}$ such that $G(v+1) \geq y$ where $G$ is defined as in \Cref{eq:w-and-g}.
        \item Set $\hat{\bw} = W(\hat{\by})$.
    \end{enumerate}
    Clearly $\hat{\bz} \overset{d}{=} \bz$ as they have the same generative process. Furthermore, in the proof of \Cref{lem:compute-w}, we argued that $\hat{\by} \overset{d}{=}  \by$ which also implies $\hat{\bw} \overset{d}{=}  \bw$. Hence, in order to apply \Cref{fact:cx-coupling} and complete the proof of this Lemma, it is sufficient to show that $\Ex[\hat{\bx} \mid \hat{\bw}] = \hat{\bw}$. Equivalently, we wish to show that
    \begin{align*}
        \Ex[\hat{\bx} \mid \hat{\by}] =  W(\hat{\by}).
    \end{align*}
    Based on our generative process, the distribution $(\hat{\ba} \mid \hat{\by})$ is given by $\Uniform(G(\hat{\by}), G(\hat{\by}+1))$. From the proof of \Cref{lem:compute-w}, we have that the expectation of $\hat{\bx} = \ba^{m-1}$ conditioned on $\ba \sim \Uniform(G(\hat{\by}), G(\hat{\by}+1))$ is exactly given by $W(\hat{\by})$, as desired.
\end{proof}

For computational reasons, we don't directly use \Cref{eq:def-W}. Instead, we use a ``capped" version of $W$. For some $y_{\max} \geq 0$
\begin{align*}
    W^{\mathrm{capped}}(y) \coloneqq W(\min(y, y_{\max})) - 10^{-5}.
\end{align*}
For concreteness, $y_{\max} = 30$ is sufficient to compute our competitive ratio. Using $W^{\mathrm{capped}}(\by)$ instead of $W$ in \Cref{lem:suff-small-1} still works as it can only lead to a smaller $\Ex_{\by \sim \D^{(1)}} \left[\left(W^{\mathrm{capped}}(\by)- t \right)_+ \right]$. We have a $y_{\max}$ to make the computation tractable. The $10^{-5}$ term is included because, without it, rounding errors make it difficult to determine if \Cref{eq:w-not-too-large} holds. Note that we also use $W^{\mathrm{capped}}$ in place of $W$ when computing the competitive ratio, so with more accurate computation (and less aggressive capping), we could slightly increase the competitive ratio. We do not believe this would change any of the first $4$ digits of accuracy which we report.

To summarize, our distribution $\W$ is defined as follows: First draw $(\by_1, \by_2, \ldots) \sim \D$, and then output $(\bx_1, \bx_2, \ldots)$ where $\bx_j = W^{\mathrm{capped}}(\by_j)$. Since this $W^{\mathrm{capped}}$ is increasing, the output of $\W$ is negatively associated. Furthermore, it is shift-invariant just as $\D$ is. 

We have verified that $\W$ satisfies the first two criteria of \Cref{thm:meta-algorithm} and the third for $r = 1$. Lastly, we verify the third for $r = 2, \ldots, m-1$ via a mix of straightforward algebra and computation. The first $r$ elements of $\mathcal{W}$ take on only finitely many values (though exponential in $r$). We fully compute a representation of this distribution for each $r = 2, \ldots, m - 1$. Given this representation, we are able to exactly compute the distribution for $\bw^{(r)}$. Verifying \Cref{eq:w-not-too-large} amounts to verifying that $g(t) \geq 0$ for all $t \in [0,1]$ where
\begin{align*}
     g(t) \coloneqq \left(1 - t  - \frac{m-r}{m} \cdot \left (1 - t^{\frac{m}{m-r}} \right) \right)- \Ex_{\bx \sim \W} \left[\left(\bw^{(r)} - t \right)_+ \right].
\end{align*}
Let $S$ be the (finitely sized) support of $\bw^{(r)}$. Then, the derivative of $g$ is easy to compute.
\begin{align*}
    g'(t) = \begin{cases}
    \Pr_{\bx \sim \W} \left[\bw^{(r)} \geq t  \right] - 1 + t^{\lfrac{r}{m-r}}& \text{if $t \notin S$}\\
    \text{undefined} & \text{otherwise}
    \end{cases}
\end{align*}
We compute all the extreme points of $g$ on the interval $[0,1]$. If $g(t) < 0$ for some $t$ on that interval, it will also be negative at some extreme point. The extreme points can only occur at the border ($t = 0$ or $t = 1$), where the derivative is undefined (on $t \in S$), or at a point where $g'(t) = 0$. Let $V$ be the set of distinct values for $\Pr_{\bx \sim \W} \left[\bw^{(r)} \geq t  \right]$, of which there are only $|S| + 1$ many. Then, if $g'(t) = 0$, it implies that, for some $v \in V$
\begin{align*}
    t = (1 - v)^{\frac{m-r}{r}}.
\end{align*}
Therefore, the total number of points we need to test for $g(t) < 0$ is finite (it has size $2 \cdot |S| + 3$). We verify that $g(t) \geq 0$ for each of those points in a publicly available Colab notebook\footnote{To see our code, go to \url{https://colab.research.google.com/drive/1yQErphKVkwwPPXsUWGT2b-nIPPaLBtnh?usp=sharing}} implying that \Cref{eq:w-not-too-large} holds.

As demonstrated in the Colab notebook, for our choice of $\W$, \Cref{eq:w-not-too-large} doesn't hold for every choice of the hyperparameter $p$. We found that for $m = 6$, it holds for all $p \in \{0.05, 0.06, \ldots, 0.95\}$, but not for $p = \{0.01, \ldots, 0.04, 0.96,\ldots, 0.99\}$. In particular, it holds for $p = 0.48$, which is the setting that maximized the resulting competitive ratio.

As a result of the above analysis and \Cref{thm:meta-algorithm}, we have constructed an $(F,m)$-discrete OCS for the $F$ defined in \Cref{eq:def f}, restated here for convenience.
\begin{align*}
    F(n) \coloneqq \Ex_{\bx \sim \W} \left[\prod_{\ell = 1}^n(1 - \bx_\ell) \right]
\end{align*}
Once again, we use computation (also in the Jupyter notebook) to determine $F(n)$. Up to $n_{\max} = 10$, we store a full representation of the first $n_{\max}$ elements from $\W$ and then use it to compute $F(n)$ for $n \leq n_{\max}$. For $n > n_{\max}$, we can upper bound $F(n) \leq F(n_{\max}) \cdot (\lfrac{F(n_{\max}) }{F(n_{\max} - 1)})^{n - n_{\max}}$ as a result of the following Lemma.

\begin{lemma}
    \label{lem:ratio-decreasing}
    For the function $F$ defined in \Cref{eq:def f}, the quantity $\frac{F(n+1)}{F(n)}$ is decreasing in $n$.
\end{lemma}
Our proof of \Cref{lem:ratio-decreasing} holds for any increasing function $W$ (so will also work for $W^{\mathrm{capped}}$), but is specialized to our choice for $\D$ from \Cref{def:D}.
\begin{proof}
    Draw $\by, \bz \sim \D(y ,z)$ and then set $\ba_k = \Ber(W(\by_k))$. We can think of $\ba_k$ as indicating whether the vertex $i$ wins the $k^{\mathrm{th}}$ time it appears. Then the following equation for $F$ is equivalent to \Cref{eq:def f}.
    \begin{align*}
        F(n) = \Prx\left[ \ba_1  =\cdots = \ba_n = 0 \right]
    \end{align*}
    Therefore,
    \begin{align*}
        \frac{F(n+1)}{F(n)} = \Prx\left[\ba_{n+1} = 0 \mid \ba_1  =\cdots = \ba_n = 0 \right].
    \end{align*}
    We want the above quantity to be less than or equal to
    \begin{align*}
        \frac{F(n)}{F(n-1)} = \Prx\left[\ba_{n+1} = 0 \mid \ba_2  =\cdots = \ba_n = 0 \right].
    \end{align*}
    using $E$ to indicate the event [$\ba_j = 0$ for each $j \in \{2, \ldots, n\}$], proving that $\frac{F(n+1)}{F(n)} \leq \frac{F(n)}{F(n-1)}$ is equivalent to
    \begin{align}
        \label{eq:expectation-less}
        \Pr[\ba_{n+1} = 0 \mid E, \ba_1 = 0] \leq \Pr[\ba_{n+1} = 0 \mid E].
    \end{align}
    For each $k = 1 ,\ldots, n$, we define the stochastic function $\bh_k:\R_{\geq 0} \to \R_{\geq 0}$,
    \begin{align*}
        \bh_k(z) = (\bz_{k+1} \mid \bz_1 = z,  \ba_1  =\cdots = \ba_{k} = 0 ).
    \end{align*}
    We claim that $\bh_k$ is a stochastically increasing function for all $k = 1, \ldots, n$. For $k = 1$ this amounts to proving following map is stochastically increasing
    \begin{align*}
        z \mapsto (\bz_2 \mid \bz_1 = z, a_1 = 0).
    \end{align*}
    Recall that by \Cref{def:D}, $\bz_2$ and $\bz_1$ are related as follows.
    \begin{enumerate}
        \item With probability equal to $p$ (the hyperparameter in \Cref{def:D}), $\bz_2 = 0$. If this happens, then $\by_1 = \bz_1$.
        \item With probability $1 - p$, $\bz_2 = \bz_1 + 1$ and then $\by_1 = -1$
    \end{enumerate} 
    If the probability that the first branch happens conditioned on $(\bz_1 = z, a_1 = 0)$ is decreasing in $z$ then $\bh_1$ is stochastically increasing. We compute that probability using Bayes theorem.
    \begin{align*}
        \Pr[\by_1 = z \mid a_1 = 0, \bz_1 = z] &= \frac{\Pr[a_1 = 0 \mid \by_1 = z , z_1 = z] \cdot \Pr[\by_1 = z  \mid \bz_1 = z]}{\Pr[a_1 = 0 \mid \bz_1 = z]} \\
        &= \frac{(1 - W(z)) }{ p\cdot(1 - W(z)) + (1 - p) \cdot (1 - W(-1))} \cdot \Pr[\by_1 = z  \mid \bz_1 = z]
    \end{align*}
    The quantity $\Pr[\by_1 = z \mid \bz_1 = z]$ is a constant (it is just equal to $p$). Since $W$ is increasing, the above is decreasing in $z$. Therefore, $\bh_1$ is stochastically decreasing.
    
    For $k \geq 2$, by the inductive hypothesis, we may that assume $\bh_{k-1}$ is stochastically increasing. Then,
    \begin{align*}
        \bh_k(z) &= (\bz_{k+1} \mid \bz_1 = z,  \ba_1  =\cdots = \ba_{k} = 0 ) \\
        &= (\bz_{k+1} \mid \bz_2 = \bz', \ba_2  =\cdots = \ba_{k} = 0) \quad \quad \text{where } \bz' \sim(\bz_2 \mid \bz_1 = z, a_1 = 0) \tag{\Cref{fact:memoryless}} \\
        &= (\bh_{k - 1} \circ \bh_1)(z) \tag{$\D$ is shift invariant}
    \end{align*}
    and since both $\bh_{k-1}$ and $\bh_1$ are stochastically increasing, so is $\bh_k$. We are now ready to prove \Cref{eq:expectation-less}. If we condition on $\ba_1 = 0$, then the distribution of $\bz_2$ is that of a $\Geo(p) + 1$. If we do not condition on $\ba_1$, that distribution is just $\Geo(p)$. Since $\bh_{n}$ is stochastically increasing and $\Geo(p) + 1$ stochastically dominates $\Geo(p)$, the distribution of $\bz_{n+1}$ conditioned on $\ba = 0$ and $E$ stochastically dominates the distribution of $\bz_{n+1}$ just conditioning on $E$. The quantity $\by_{n+1}$ is a stochastically increasing function of $\bz_{n+1}$ and the quantity $\ba_{n+1}$ is a stochastically increasing function of $\by_{n+1}$. Therefore, the probability $\ba_{n+1} = 0$ conditioned on $E$ and $\ba_1 = 0$ is lower than just conditioning on $E$. \Cref{eq:expectation-less} holds and therefore the desired result does as well.
\end{proof}

As a consequence of the above, we have proven the existence of an $(F,6)$-OCS for the $F$ in \Cref{fig:OCS-parameters}.

\section{Proof of \Cref{lem:negative}}
\label{sec:neg}
\negative*

\begin{proof}
    We prove that there is no $2$-discrete OCS satisfying both of the following.
    \begin{enumerate}
        \item The probability a vertex is picked in a single time step it appears is $\frac{1}{2}$
        \item Let $j_1, j_2$ be two time steps vertex $i$ appears in for which $i$ never appears in a time step between $j_1$ and $j_2$. Then, the probability $i$ is selected in $j_1$ or $j_2$ is strictly greater than $\frac{5}{6}$.
    \end{enumerate}
    This implies \Cref{lem:negative}. Suppose, for the sake of contradiction, such an OCS exists. Feed in the following three pairs into the OCS. $((a,b), (b,c), (a,b))$. In the first round, the OCS must pick $a$ and $b$ each with probability $\frac{1}{2}$. 
    
    Then, in order to guarantee that $b$ is picked in one of the first two time steps with probability more than $\frac{5}{6}$, if $a$ is picked in the first round, $b$ must be picked in the second round with probability $p_2 > \lfrac{2}{3}$.
    
    We also need the probability that $c$ is picked in the second round to be $\frac{1}{2}$. To guarantee this, in the case where $b$ is picked in the first round, we need $c$ to be picked with probability $p_2$. Hence, there is a $\frac{1}{2} \cdot p_2$ probability that $b$ is picked in the first round and $c$ in the second. Let $p_3$ be the probability that $a$ is picked in this case. If $p_3 \geq \lfrac{1}{2}$, then there is a
    \begin{align*}
        \frac{1}{2} \cdot p_2 \cdot p_3 > \frac{1}{6}
    \end{align*}
    chance that $b$ is not selected in the second or third round. On the other hand, if $p_3 \leq \lfrac{1}{2}$, then there is a
    \begin{align*}
        \frac{1}{2} \cdot p_2 \cdot (1 - p_3) > \frac{1}{6}
    \end{align*}
    chance that $a$ is not selected in the first or third round. In both cases, we have a contradiction.
\end{proof}

\bibliography{OCS}{}
\bibliographystyle{alpha}

\appendix

\section{Missing proofs from \Cref{sec:OCS-motivation}}
\label{sec:proofs-OCS-motivation}

\subsection{Proof that independent randomness is not enough}

We prove the that algorithms which make independent selections at each time step cannot achieve a competitive ratio better than that of the greedy algorithm.
\factindependent*

\begin{proof}
    For each $n \in \Z_{\geq 1}$, we'll show that no algorithm using independent randomness can acheive a competitive ratio better than $\frac{1}{2} + \frac{1}{2n}$ on every graph with $n$ offline vertices. Fix some $n \in \Z_{\geq 1}$ and let $\mathcal{G}$ be the family of triangular graphs with vertices $L = R = [n]$ defined as follows: For each permutation $\pi$ of $[n]$, there is one graph $G \in \mathcal{G}$. The first online vertex has an edge to every offline vertex. Then, for each online vertex where $j \geq 2$ has an edge to every offline vertex that $j - 1 \in R$ has an edge to \emph{except for} $\pi_{j-1}$. 
    
    It's not too hard to see that each $G \in \mathcal{G}$ has a perfect matching of all $n$ edges formed by matching $j$ to $\pi_j$. We'll prove that for any algorithm using independent randomness, there is some $G \in \mathcal{G}$ on which that algorithm has an expected number of matches at most $\frac{n+1}{2}$.
    
    Let $\mathcal{A}$ be some algorithm that uses independent randomness. At time step $j$, as a deterministic function of the portion of the graph revealed so far, the algorithm decides on a distribution over $L$ on which to match the $j^{\text{th}}$ vertex to. Then, it independently picks a match for $j$ independent of all previous matches. In particular, if it picks a vertex in $L$ that has been picked by some $j' < j$, this step is wasted. We'll use $\pj$ to represent the distribution over $L$ that $\mathcal{A}$ chooses for $j \in R$, where $(\pj)_i$ is the probability that $j$ is matched to $i$.
    
    We will adversarial choose a $G \in \mathcal{G}$ depending on the algorithm $\mathcal{A}$. Once the $j^{\text{th}}$ vertex arrives, the edges of each $j' \in R$ for $j' \leq j$ are known. This corresponds to knowing $\pi_1, \ldots, \pi_{j - 1}$. Then, $\mathcal{A}$ must choose $\pj$, and we are free to adversarially choose $\pi_j$ as it only affects the portion of the graph not yet revealed. Recall that whichever vertex we select as $\pi_j$ will be unable to be matched in the future, as it will not have edges to any $j' \in R$ where $j' > j$. We'll set $\pi_j$ to the offline vertex not yet removed that is \emph{least likely} to be already matched. Formally,
    \begin{align}
        \label{eq:select-pi}
        \pi_j = \argmin_{i \in L, i \neq \pi_{j'} \text{ for any $j' < j$}} \left(1 - \prod_{j' \leq j}(1 -(p^{(j')})_i)  \right).
    \end{align}
    
    Next, we show that for any $\mathcal{A}$ using independent randomness, that on this adversarial choice of $G \in \mathcal{G}$, $\mathcal{A}$ matches at most $\frac{n+1}{2}$ edges in expectation. For each $k \in [n]$, we define $R_k$ to be the sum over the last $k$ offline nodes to be removed (the nodes $\pi_{n - k + 1}, \ldots, \pi_{n}$), that each is \emph{not} matched to one of the first $n - k$ online nodes. Formally,
    \begin{align*}
        R_k \coloneqq  \sum_{i \in \{\pi_{n - k + 1}, \ldots, \pi_{n}\}}\left[\text{$i$ is not matched to $j$} \in \{\pi_1, \ldots, \pi_{n - k}\} \right]
    \end{align*}
    We claim, and will prove by induction, that the expected number of matches made to the last $k$ offline nodes (the nodes $\pi_{n - k + 1}, \ldots, \pi_{n}$) is at most $k  - R_k \cdot \frac{k-1}{2k}$. As $R_n = n$, this immediately implies the expected number of matches in the entire graph is at most $\frac{n+1}{2}$, the desired result.
    
    For our base case, when $k = 1$, we trivially have the expected number of matches of the last vertex is at most $1$. For the inductive case, fix some $k \geq 2$.
    We consider the state after $n - k$ rounds have occurred. For each $\ell \in [k]$, let $r_\ell$ be the probability that $\pi_{n - \ell+ 1}$ is not matched in any previous round (to $j \in [n - k]$) and
    $q_\ell \coloneqq 1 - (\pj)_{\pi_{n - \ell + 1}}$ be the probability it is not matched in the next round (to $j = n - k + 1$), given that it was not matched in previous rounds. As adversarially chosen in \Cref{eq:select-pi}, $\pi_{n - k + 1}$ will be set to whichever node maximizes $r_\ell \cdot q_\ell$. Let $\ell^\star$ be the $\ell$ maximizing $r_\ell \cdot q_\ell$. Then, $R_{k - 1}  = \sum_{\ell \neq \ell^\star} r_\ell \cdot q_\ell$. The expected number of matches to the last $k$ offline vertices is at most
    \begin{align*}
        f(r,q) \coloneqq \underbrace{1 - r_{\ell^\star}q_{\ell^\star}}_{\text{\clap{$\Pr[\pi_{n - k + 1}\text{ matched}]$}~}} \quad + \quad\underbrace{k - 1 - \frac{k-2}{2(k-1)} \cdot \sum_{\ell \neq \ell^\star}r_{\ell} \cdot q_{\ell}}_{\text{\clap{$\Ex[\text{matches to $\pi_{n - k + 2}, \ldots, \pi_n$}]$}}} \quad\quad \text{where }\ell^\star = \text{arg}\max_{\ell \in [k]} r_\ell \cdot q_\ell.
    \end{align*}
    We used the inductive hypothesis to bound $\Ex[\text{matches to $\pi_{n - k + 2}, \ldots, \pi_n$}]$. Furthermore, $r_\ell$ and $q_\ell$ obey some bounds
    \begin{align}
        \label{eq:r-q-constraints}
        \begin{aligned}
        \sum_{\ell \in [k]} r_\ell &= R_k &&\text{(Definition of $R_k$)}\\
        \sum_{\ell \in [k]} q_{\ell} &= \sum_{\ell \in [k]}1 - (\pj)_{\pi_{n - \ell + 1}} = k-1 \\
        r_\ell, q_\ell &\geq 0 &&\text{(for all $\ell \in [k]$)}
        \end{aligned}
    \end{align}
    The constraints of \Cref{eq:r-q-constraints} cause $r_1, \ldots, r_k$ and $q_1, \ldots, q_k$ to lie in a closed and bounded space, and $f$ is continuous. Therefore, $f$ attains its maximum value in this space. We may assume that $R_k > 0$, as if $R_k = 0$, $f$ is $k-1$ everywhere satisfying the above constraints. Similarly, as we only care about cases when $k \geq 2$, the sum of $q_\ell$ is also strictly more than $0$. We claim the maximum of $f$ subject to constraints of \Cref{eq:r-q-constraints} occurs when
    \begin{align}
        \label{eq:r-q-max}
        r_1 = \cdots = r_k = \frac{R_k}{k} \quad \text{and} \quad q_1 = \cdots = q_k = \frac{k - 1}{k}.
    \end{align}
    Consider any other point, ($r_1, \ldots, r_k, q_1, \ldots, q_k$), satisfying the constraints of  \Cref{eq:r-q-constraints}. Using the notation $\frac{\partial f}{\partial x_+}$ to the refer to the \emph{right} derivative, how much $f$ changes as $x$ is increased by an $\eps$ amount, and $\frac{\partial f}{\partial x_-}$ for the left derivatives,
    \begin{enumerate}
        \item If $r_\ell \cdot q_\ell < r_{\ell'}\cdot q_{\ell'}$ for all $\ell' \neq \ell$, then,
        \begin{align*}
            &\frac{\partial f}{\partial r_\ell} = - \frac{k-2}{2(k-1)}\cdot q_i & \frac{\partial f}{\partial q_\ell} = - \frac{k-2}{2(k-1)} \cdot r_i.
        \end{align*}
        We note that the term $\frac{k-2}{2(k-1)}$ is in $[0,\frac{1}{2})$ for $k \geq 2$. In particular, it is never more than $1$, and strictly more than $0$ whenever $k \geq 3$. When $k = 2$, it is easy to show the maximum of $f$ subject to the constraints of \Cref{eq:r-q-constraints} occurs at the point in \Cref{eq:r-q-max}. For the remainder of this proof, we consider the $k \geq 3$ case where $0 <  \frac{k-2}{2(k-1)} < \frac{1}{2}$.
        \item If $r_\ell \cdot q_\ell > r_{\ell'} \cdot q_{\ell'}$ for all $\ell' \neq \ell$, then,
        \begin{align*}
            &\frac{\partial f}{\partial r_\ell} = -  q_i & \frac{\partial f}{\partial q_\ell} = -  r_i.
        \end{align*}
        \item If  $r_\ell \cdot q_\ell \geq r_{\ell'} \cdot q_{\ell'}$ for all $\ell' \neq \ell$, and there is at least one $\ell' \neq \ell$ where $r_\ell \cdot q_\ell = r_{\ell'} \cdot q_{\ell'}$, then
        \begin{align*}
            &\frac{\partial f}{\partial (r_\ell)_+} =  -  q_i & \frac{\partial f}{\partial (q_\ell)_+} = -  r_i.\\
            &\frac{\partial f}{\partial (r_\ell)_-} = - \frac{k-2}{2(k-1)}\cdot q_i & \frac{\partial f}{\partial (q_\ell)_-} = - \frac{k-2}{2(k-1)}\cdot r_i.
        \end{align*}
    \end{enumerate}
    Consider some point ($r_1, \ldots, r_k, q_1, \ldots, q_k$) satisfying the constraints of \Cref{eq:r-q-constraints} that is not the point in \Cref{eq:r-q-max}. Then, there exists some $a, b$ where $r_a \cdot q_a \geq r_{\ell} \cdot q_{\ell}$ for each $\ell \neq a$, and either $r_a \neq r_b$ or $q_a \neq q_b$. Then, either $r_a > r_b$ or $q_a > q_b$. Without loss of generality, we assume $r_a > r_b$. Then, for sufficiently small $\eps$, if we decrease $q_a \leftarrow q_a - \eps$ and increase $q_b \leftarrow a_b + \eps$, the value of $f$ will increase and the constraints remain satisfied. Hence, the maximum value of $f$ is
    \begin{align*}
        \max_{r,q}f(r,q) &= 1 - \frac{R_k}{k}\frac{k - 1}{k} + k - 1 - \frac{k-2}{2(k-1)} \cdot (k-1) \cdot \frac{R_k}{k} \cdot \frac{k - 1}{k} \\
        &= k - R_k \cdot \frac{k-1}{2k}.
    \end{align*}
    Therefore, by induction, the maximum expected matches made to the last $k$ offline nodes is at most $k - R_k \cdot \frac{k-1}{2k}$. Plugging in $k = n$ gives the desired result.
\end{proof}

\subsection{Proof that it is impossible to negatively correlate all time steps}

We prove that there does not exists a continuous OCS that negatively correlated every single time step. This is why we choose to have our OCS only negatively correlate time steps that are close.
\begin{lemma}
    \label{lem:no-neg-always}
    Let $f:\Rpos \to [0,1]$ be a function on which there exists some rational $x \in \Q$ where $f(x) < e^{-x}$. Then, there does not exist a continuous OCS with the following guarantee: For any element $i \in L$ and set of time steps $S$,
    \begin{align*}
        \Pr[\bi_j \neq i \text{ for every }j \in S] \leq  f\left(\sum_{j \in S} (\pj)_i \right).
    \end{align*}
\end{lemma}
Note that $f(x) = e^{-x}$ corresponds to an OCS which picks an winner independently at each time step, and negative correlation corresponds to $f(x) < e^{-x}$.
\begin{proof}
    Let $x = \frac{a}{b}$ for $a,b \in Z$. For $k,n \in Z_{\geq 1}$, consider the input instance where $|L| = b \cdot k$ and each of $n$ time steps have $\pj$ uniform over $L$, meaning $(\pj)_i = \frac{1}{b\cdot k}$ for all $i \in L, j \in [n]$. We will show that in the limit of $\frac{n}{k} \to \infty$ and $k \to \infty$, the desired guarantee is impossible.
    
    Assume, for the sake of guarantee, the desired result is possible. We will show that this implies that the expected number of total winners is more than $n$, a contradiction as each time step has only a single winner. Fix an arbitrary $i \in V$, and let $\bz_j$ be the event that $i$ is \emph{not} the winner in the $j^{\text{th}}$ round. Then, by Jensen's inequality,
    \begin{align*}
        \Ex\left[\sum_{j \in [n]} \bz_j\right] \leq \sqrt[ak]{\Ex\left[\left(\sum_{j \in [n]} \bz_j\right)^{ak}\right]}.
    \end{align*}
    Expanding $(\sum_{j \in [n]} \bz_j)^{ak}$, there are $\frac{n!}{(n-ak)!}$ terms which are the product of $ak$ distinct $\bz_{j_1}, \ldots, \bz_{j_{ak}}$, and $n^{ak} - \frac{n!}{(n-ak)!}$ which are the product of non-distinct events. Whenever $\bz_{j_1}, \ldots, \bz_{j_k}$ are non-distinct, we use the lose upper bound that that $E[\bz_{j_1} \cdot \cdots \cdot \bz_{j_{ak}}] \leq 1$. Whenever they are distinct, by the assumption on the quality of the OCS, it must be the case that $E[\bz_{j_1} \cdot \cdots \cdot \bz_{j_{ak}}] \leq f(x)$. Therefore,
    \begin{align*}
        \Ex\left[\sum_{j \in [n]} \bz_j\right] &\leq \sqrt[ak]{\frac{n!}{(n-ak)!} \cdot f(x) + n^{ak} - \frac{n!}{(n-ak)!}} \\
        &\leq\sqrt[ak]{\frac{n!}{(n-ak)!}} \cdot \sqrt[ak]{f(x)} + \sqrt[ak]{n^{ak} - \frac{n!}{(n-ak)!}}
    \end{align*}
    In the limit of $\frac{n}{k} \to \infty$,
    \begin{align*}
         \Ex\left[\sum_{j \in [n]} \bz_j\right] &\leq n \cdot (1 - o(1) \cdot \sqrt[ak]{f(x)} \\
         &< n \cdot(1 - o(1)) \cdot  \sqrt[ak]{e^{-\frac{ak}{bk}}} \\
         &= n \cdot(1 - o(1)) \cdot e^{-\frac{1}{bk}}
    \end{align*}
    Note that the expected number of times that $i$ is the winner is $n - \Ex\left[\sum_{j \in [n]} \bz_j\right]$. In the limit of $k \to \infty$, this value is strictly more than $n \cdot \frac{1}{bk}$. By symmetry, this is true of all $bk$ different vertices in $L$, implying, the total number of winner is more than $n$, a contradiction.
\end{proof}
\section{Missing proofs from \Cref{sec:continuous-OCS-formal}}
\label{sec:prove-a-properties}
\aproperties*
\begin{proof}[Proof of \Cref{eq:a diffeq}]
    Using integration by parts,
    \begin{align*}
        a'(x) &= f'(x) - \int_{0}^\infty e^{-t}f'(t + x)dt \\
        a'(x) &= f'(x) - \left(\bigg[e^{-t} f(t+x)\bigg]_{t = 0}^{t = \infty} + \int_{0}^\infty e^{-t}f(t + x)dt  \right) \tag{integration by parts} \\
        a'(x) &= f'(x) - \left(-f(x) + \int_{0}^\infty e^{-t}f(t + x)dt  \right) \\
        &= f'(x) + a(x).
    \end{align*}
\end{proof}

\begin{proof}[Proof of \Cref{eq:a-decr}]
    Using the convexity of $f$,
    \begin{align*}
        a'(x) &=  f'(x) - \int_{0}^\infty e^{-t}f'(t + x)  \\
        &\leq f'(x) - \int_{0}^\infty e^{-t}f'( x)  \tag{$f$ is convex}\\
        &=0
    \end{align*}

\end{proof}

\begin{proof}[Proof of \Cref{eq:a-gamma}]
    This is immediate from the fact that $f(0) = 1$ and the definitions of $a$ in \Cref{eq:def-a} and of $\Gamma$ in \Cref{eq:def-Gamma}.
\end{proof}
We'll use the following proposition.
\begin{proposition}
    \label{prop:log-concave-integral}
    If $f:\Rpos \to \Rpos$ is log-concave, then
    \begin{align*}
        \int_{0}^\infty \frac{e^{-t}f(t+x)}{f(x)}dt
    \end{align*}
    is decreasing in $x$.
\end{proposition}
\begin{proof}
   It is enough for $\frac{f(t+x)}{f(x)}$ to be decreasing in $x$ for any fixed $t \geq 0$. The derivative of this expression is
   \begin{align*}
       \frac{d}{dx}\left(\frac{f(t+x)}{f(x)} \right) = \frac{f(x)f'(t+x) - f(t+x)f'(x)}{f(x)^2}.
   \end{align*}
   Hence, the desired expression is decreasing when
   \begin{align*}
       \frac{f'(x)}{f(x)} \geq \frac{f'(t+x)}{f(t+x)}.
   \end{align*}
   $f$ is log-concave iff $\frac{f'(x)}{f(x)}$ is decreasing in $x$. Therefore, the desired result follows by the fact that $t \geq 0$ and the log-concavity of $f$.
   
\end{proof}

\begin{proof}[Proof of \Cref{eq: f lower a}]
   If $f(x) = 0$, the desired result holds. We wish to show that, for all $x \geq 0$ where $f(x) > 0$,
   \begin{align*}
       \frac{f(x) \cdot \Gamma}{a(x)} \leq 1.
   \end{align*}
   As $f(0) = 1$ and $a(0) = \Gamma$, the above expression is $1$ when $x = 0$. Therefore, it is enough to show that expression is decreasing. That's equivalent to showing that $\frac{f(x)}{a(x)}$ is decreasing, which, since we only need to worry about the case where $f(x) > 0$, is equivalent to showing that $\frac{a(x)}{f(x)}$ is increasing.
   \begin{align*}
       \frac{a(x)}{f(x)} &= \frac{f(x) - \int_{0}^\infty e^{-t} f(t + x) dt}{f(x)} \\
       &= 1 - \int_{0}^\infty \frac{e^{-t}f(t+x)}{f(x)}dt
   \end{align*}
   By \Cref{prop:log-concave-integral}, the above is increasing as a consequence of the log-concavity of $f$.
\end{proof}

In order to prove \Cref{eq: r upper bound}, we'll need the following proposition.

\begin{proposition}
    \label{prop:fractional-diffeq}
    For any continuous $f,g: \Rpos \to \Rpos$ where $g > 0$ satisfying
    \begin{align*}
        \frac{f(x)}{g(x)} \quad\quad \text{ is decreasing for all $x \geq 0$}.
    \end{align*}
    Let $F(x) \coloneqq \int_0^x f(t)dt$ and $G(x) \coloneqq \int_0^x g(t)dt$. Then, the following is also decreasing for all $x \geq 0$,
    \begin{align*}
        \frac{F(x)}{G(x)}.
    \end{align*}
\end{proposition}
\begin{proof}
   We compute the derivative of the desired expression:
   \begin{align*}
       \frac{d}{dx}\frac{F(x)}{G(x)} = \frac{G(x)f(x) - F(x)g(x)}{G(x)^2}
   \end{align*}
   With a bit of rearranging, we have that the following are equivalent for all $x > 0$.
   \begin{align*}
       \frac{d}{dx}\frac{F(x)}{G(x)} < 0 \quad\quad\quad \iff \quad\quad\quad z(x) \coloneqq \frac{f(x)}{g(x)} - \frac{F(x)}{G(x)} < 0 \\
       \frac{d}{dx}\frac{F(x)}{G(x)} = 0 \quad\quad\quad \iff \quad\quad\quad z(x) \coloneqq \frac{f(x)}{g(x)} - \frac{F(x)}{G(x)} = 0 \\
        \frac{d}{dx}\frac{F(x)}{G(x)} > 0 \quad\quad\quad \iff \quad\quad\quad z(x) \coloneqq \frac{f(x)}{g(x)} - \frac{F(x)}{G(x)} > 0
   \end{align*}
   Also, by L' H\^{o}pital's rule, $\lim_{x \to 0} z(x) = 0$. By continuity of $z$, if $z(x) > 0$ for some $x > 0$, it must imply for some earlier $x^\star \in (0, x)$ that $z(x^\star) \geq 0$ and $z'(x^\star) > 0$. However, as $\frac{f(x)}{g(x)}$ is decreasing, if $z'(x^\star) > 0$, it must be the case that $\frac{F(x)}{G(x)}$ is decreasing at $x^\star$. This contradicts that $z(x^\star) \geq 0$. Therefore, $z(x) \leq 0$ for all $x > 0$ and the desired expression is decreasing.
\end{proof}

\begin{proof}[Proof of \Cref{eq: r upper bound}]
    Combining \Cref{eq:a-decr,eq:a-gamma}, we have that $a(x) \leq \Gamma$ for all $x \geq 0$. Therefore, if $r = 0$, we are done. Otherwise, we wish to show that, for all $x \geq 0$,
    \begin{align*}
        \frac{1}{r} \geq h(x) \coloneqq \frac{a(x) - f(x) \cdot \Gamma}{\Gamma - a(x)}
    \end{align*}
    We will prove that $h(x)$ is decreasing. Assuming that is true, the desired result follows because
    \begin{align*}
        h(x) \leq \lim_{x \to 0} h(x) &= \frac{a'(0) - f'(0) \cdot \Gamma}{-a'(0)} \tag{L' H\^{o}pital's rule}\\
        &= \frac{f'(0) + a(0) - f'(0) \cdot \Gamma}{-(f'(0) + a(0))} \tag{\Cref{eq:a diffeq}} \\
        &= \frac{\Gamma + (1 - \Gamma) \cdot f'(0) }{-f'(0) - \Gamma} \tag{\Cref{eq:a-gamma}} \\
        &\leq \frac{1}{r} \tag{\Cref{eq:r-upper-bound-thm}}
    \end{align*}
    Let us define $A(x) \coloneqq \frac{a(x)}{\Gamma}$. Then
    \begin{align*}
        h(x) &= \frac{A(x) - f(x)}{1 - A(x)} \\
        &= \frac{(1 - f(x)) - (1 - A(x))}{1 - A(x)} \\
        &= \frac{1 - f(x)}{1 - A(x)} - 1 \\
       &=\frac{\int_{0}^x f'(t)dt}{\int_0^x A'(t)dt} - 1  \tag{f(0) = A(0) = 1}
    \end{align*}
    By \Cref{prop:fractional-diffeq}, in order to show $h(x)$ is decreasing, it is enough to show that $\frac{f'(x)}{A'(x)}$ is decreasing. Recall that $A(x) \coloneqq \frac{a(x)}{\Gamma}$, so equivalently, we can show that $\frac{f'(x)}{a'(x)}$ is decreasing. Using the fact that $a'(x)$ and $f'(x)$ never change signs, we can instead show that $\frac{a'(x)}{f'(x)}$ is increasing.
    \begin{align*}
        \frac{a'(x)}{f'(x)} &= \frac{f'(x) - \int_{0}^\infty e^{-t}f'(t + x)dt}{f'(x)} \tag{\Cref{eq:def-a}} \\ 
        &=  1 - \int_{0}^\infty \frac{e^{-t}f'(t + x)}{f'(x)}dt
    \end{align*}
    By \Cref{prop:log-concave-integral} and the log concavity of $f'$, the above is increasing, completing the proof of \Cref{eq: r upper bound}.

\end{proof}

\end{document}